\documentclass[a4paper,11pt]{article}
\usepackage[T1]{fontenc} 

\makeatletter
\gdef\@fpheader{ }
\gdef\@journal{ }
\makeatother
\RequirePackage{amsmath}
\RequirePackage{amssymb}
\RequirePackage{epsfig}
\RequirePackage{graphicx}
\RequirePackage[numbers,sort&compress]{natbib}
\RequirePackage{color}
\RequirePackage[colorlinks=true
,urlcolor=blue
,anchorcolor=blue
,citecolor=blue
,filecolor=blue
,linkcolor=blue
,menucolor=blue
,pagecolor=blue
,linktocpage=true
,pdfproducer=medialab
,pdfa=true
]{hyperref}

\newif\ifnotoc\notocfalse
\newif\ifemailadd\emailaddfalse
\newif\iftoccontinuous\toccontinuousfalse
\makeatletter
\def\@subheader{\@empty}
\def\@keywords{\@empty}
\def\@abstract{\@empty}
\def\@xtum{\@empty}
\def\@dedicated{\@empty}
\def\@arxivnumber{\@empty}
\def\@collaboration{\@empty}
\def\@collaborationImg{\@empty}
\def\@proceeding{\@empty}
\def\@preprint{\@empty}

\newcommand{\subheader}[1]{\gdef\@subheader{#1}}
\newcommand{\keywords}[1]{\if!\@keywords!\gdef\@keywords{#1}\else%
\PackageWarningNoLine{\jname}{Keywords already defined.\MessageBreak Ignoring last definition.}\fi}
\renewcommand{\abstract}[1]{\gdef\@abstract{#1}}
\newcommand{\dedicated}[1]{\gdef\@dedicated{#1}}
\newcommand{\arxivnumber}[1]{\gdef\@arxivnumber{#1}}
\newcommand{\proceeding}[1]{\gdef\@proceeding{#1}}
\newcommand{\xtumfont}[1]{\textsc{#1}}
\newcommand{\correctionref}[3]{\gdef\@xtum{\xtumfont{#1} \href{#2}{#3}}}
\newcommand\jname{JHEP}
\newcommand\acknowledgments{\section*{Acknowledgments}}

\newcommand\preprint[1]{\gdef\@preprint{\hfill #1}}

\newtheorem{theorem}{Theorem}

\newenvironment{proof}[1][Proof]{\noindent\textbf{#1.} }{\ \rule{0.5em}{0.5em}}

\makeatother

\newcommand\note[2][]{%
\if!#1!%
\stepcounter{footnote}\footnotetext{#2}%
\else%
{\renewcommand\thefootnote{#1}%
\footnotetext{#2}}%
\fi}

\makeatletter
\newtoks\auth@toks
\renewcommand{\author}[2][]{%
  \if!#1!%
    \auth@toks=\expandafter{\the\auth@toks#2\ }%
  \else
    \auth@toks=\expandafter{\the\auth@toks#2$^{#1}$\ }%
  \fi
}
\makeatother
\makeatletter
\newtoks\affil@toks\newif\ifaffil\affilfalse
\newcommand{\affiliation}[2][]{%
\affiltrue
  \if!#1!%
    \affil@toks=\expandafter{\the\affil@toks{\item[]#2}}%
  \else
    \affil@toks=\expandafter{\the\affil@toks{\item[$^{#1}$]#2}}%
  \fi
}
\makeatother
\makeatletter
\newtoks\email@toks\newcounter{email@counter}%
\newcommand{\emailAdd}[1]{%
\emailaddtrue%
\ifnum\theemail@counter>0\email@toks=\expandafter{\the\email@toks, \@email{#1}}%
\else\email@toks=\expandafter{\the\email@toks\@email{#1}}%
\fi\stepcounter{email@counter}}
\newcommand{\@email}[1]{\href{mailto:#1}{\tt #1}}
\makeatother

\makeatletter
\newcommand*\collaboration[1]{\gdef\@collaboration{#1}}
\newcommand*\collaborationImg[2][]{\gdef\@collaborationImg{#2}}
\makeatletter
\newcommand\afterLogoSpace{\smallskip}
\newcommand\afterSubheaderSpace{\vskip3pt plus 2pt minus 1pt}
\newcommand\afterProceedingsSpace{\vskip21pt plus0.4fil minus15pt}
\newcommand\afterTitleSpace{\vskip23pt plus0.06fil minus13pt}
\newcommand\afterRuleSpace{\vskip23pt plus0.06fil minus13pt}
\newcommand\afterCollaborationSpace{\vskip3pt plus 2pt minus 1pt}
\newcommand\afterCollaborationImgSpace{\vskip3pt plus 2pt minus 1pt}
\newcommand\afterAuthorSpace{\vskip5pt plus4pt minus4pt}
\newcommand\afterAffiliationSpace{\vskip3pt plus3pt}
\newcommand\afterEmailSpace{\vskip16pt plus9pt minus10pt\filbreak}
\newcommand\afterXtumSpace{\par\bigskip}
\newcommand\afterAbstractSpace{\vskip16pt plus9pt minus13pt}
\newcommand\afterKeywordsSpace{\vskip16pt plus9pt minus13pt}
\newcommand\afterArxivSpace{\vskip3pt plus0.01fil minus10pt}
\newcommand\afterDedicatedSpace{\vskip0pt plus0.01fil}
\newcommand\afterTocSpace{\bigskip\medskip}
\newcommand\afterTocRuleSpace{\bigskip\bigskip}
\newlength{\affiliationsSep}
\newcommand\beforetochook{\pagestyle{myplain}\pagenumbering{roman}}

\DeclareFixedFont\trfont{OT1}{phv}{b}{sc}{11}

\renewcommand\maketitle{
\pagestyle{empty}
\thispagestyle{titlepage}
\setcounter{page}{0}
\noindent{\small\scshape\@fpheader}\@preprint\par

\afterLogoSpace
\if!\@subheader!\else\noindent{\trfont{\@subheader}}\fi
\afterSubheaderSpace
\if!\@proceeding!\else\noindent{\sc\@proceeding}\fi
\afterProceedingsSpace
{\LARGE\flushleft\sffamily\bfseries\@title\par}
\afterTitleSpace
\hrule height 1.5\p@%
\afterRuleSpace
\if!\@collaboration!\else
{\Large\bfseries\sffamily\raggedright\@collaboration}\par
\afterCollaborationSpace
\fi
\if!\@collaborationImg!\else
{\normalsize\bfseries\sffamily\raggedright\@collaborationImg}\par
\afterCollaborationImgSpace
\fi
{\bfseries\raggedright\sffamily\the\auth@toks\par}
\afterAuthorSpace
\ifaffil\begin{list}{}{%
\setlength{\leftmargin}{0.28cm}%
\setlength{\labelsep}{0pt}%
\setlength{\itemsep}{\affiliationsSep}%
\setlength{\topsep}{-\parskip}}
\itshape\small%
\the\affil@toks
\end{list}\fi
\afterAffiliationSpace
\ifemailadd 
\noindent\hspace{0.28cm}\begin{minipage}[l]{.9\textwidth}
\begin{flushleft}
\textit{E-mail:} \the\email@toks
\end{flushleft}
\end{minipage}
\else 
\PackageWarningNoLine{\jname}{E-mails are missing.\MessageBreak Plese use \protect\emailAdd\space macro to provide e-mails.}
\fi
\afterEmailSpace
\if!\@xtum!\else\noindent{\@xtum}\afterXtumSpace\fi
\if!\@abstract!\else\noindent{\renewcommand\baselinestretch{.9}\textsc{Abstract:}}\ \@abstract\afterAbstractSpace\fi
\if!\@keywords!\else\noindent{\textsc{Keywords:}} \@keywords\afterKeywordsSpace\fi
\if!\@arxivnumber!\else\noindent{\textsc{ArXiv ePrint:}} \href{http://arxiv.org/abs/\@arxivnumber}{\@arxivnumber}\afterArxivSpace\fi
\if!\@dedicated!\else\vbox{\small\it\raggedleft\@dedicated}\afterDedicatedSpace\fi
\ifnotoc\else
\iftoccontinuous\else\newpage\fi
\beforetochook\hrule
\tableofcontents
\afterTocSpace
\hrule
\afterTocRuleSpace
\fi
\setcounter{footnote}{0}
\pagestyle{myplain}\pagenumbering{arabic}
} 

\renewcommand{\baselinestretch}{1.1}\normalsize
\divide\@tempdima\baselineskip
\@tempcnta=\@tempdima
\skip\@mpfootins = \skip\footins
\renewcommand{\@dotsep}{10000}

\newcommand\ps@myplain{
\pagenumbering{arabic}
\renewcommand\@oddfoot{\hfill-- \thepage\ --\hfill}
\renewcommand\@oddhead{}}
\let\ps@plain=\ps@myplain

\newcommand\ps@titlepage{\renewcommand\@oddfoot{}\renewcommand\@oddhead{}}

\numberwithin{equation}{section}

\renewcommand\section{\@startsection{section}{1}{\z@}%
                                   {-3.5ex \@plus -1.3ex \@minus -.7ex}%
                                   {2.3ex \@plus.4ex \@minus .4ex}%
                                   {\normalfont\large\bfseries}}
\renewcommand\subsection{\@startsection{subsection}{2}{\z@}%
                                   {-2.3ex\@plus -1ex \@minus -.5ex}%
                                   {1.2ex \@plus .3ex \@minus .3ex}%
                                   {\normalfont\normalsize\bfseries}}
\renewcommand\subsubsection{\@startsection{subsubsection}{3}{\z@}%
                                   {-2.3ex\@plus -1ex \@minus -.5ex}%
                                   {1ex \@plus .2ex \@minus .2ex}%
                                   {\normalfont\normalsize\bfseries}}
\renewcommand\paragraph{\@startsection{paragraph}{4}{\z@}%
                                   {1.75ex \@plus1ex \@minus.2ex}%
                                   {-1em}%
                                   {\normalfont\normalsize\bfseries}}
\renewcommand\subparagraph{\@startsection{subparagraph}{5}{\parindent}%
                                   {1.75ex \@plus1ex \@minus .2ex}%
                                   {-1em}%
                                   {\normalfont\normalsize\bfseries}}

\def\fnum@figure{\textbf{\figurename\nobreakspace\thefigure}}
\def\fnum@table{\textbf{\tablename\nobreakspace\thetable}}

\long\def\@makecaption#1#2{%
  \vskip\abovecaptionskip
  \sbox\@tempboxa{\small #1. #2}%
  \ifdim \wd\@tempboxa >\hsize
    \small #1. #2\par
  \else
    \global \@minipagefalse
    \hb@xt@\hsize{\hfil\box\@tempboxa\hfil}%
  \fi
  \vskip\belowcaptionskip}

\renewenvironment{thebibliography}[1]{%
\begin{oldthebibliography}{#1}%
\small%
\raggedright%
\setlength{\itemsep}{5pt plus 0.2ex minus 0.05ex}%
}%
{%
\end{oldthebibliography}%
}

\usepackage{titlesec} 
\usepackage{titletoc} 

\begin{document}



\renewcommand{\thefootnote}{\fnsymbol{footnote}}

\title{\boldmath A statistical mechanical approach to restricted integer partition functions}%


\author[a]{Chi-Chun Zhou}
\author[a,*]{and Wu-Sheng Dai}\note{daiwusheng@tju.edu.cn.}


\affiliation[a]{Department of Physics, Tianjin University, Tianjin 300350, P.R. China}









\abstract{
The main aim of this paper is twofold: (1) Suggesting a statistical mechanical
approach to the calculation of the generating function of restricted integer
partition functions which count the number of partitions ----- a way of
writing an integer as a sum of other integers under certain restrictions. In
this approach, the generating function of restricted integer partition
functions is constructed from the canonical partition functions of various
quantum gases. (2) Introducing a new type of restricted integer partition
functions corresponding to general statistics which is a generalization of
Gentile statistics in statistical mechanics; many kinds of restricted integer
partition functions are special cases of this restricted integer partition
function. Moreover, with statistical mechanics as a bridge, we reveals a
mathematical fact: the generating function of restricted integer partition
function is just the symmetric function which is a class of functions being
invariant under the action of permutation groups. Using the approach, we
provide some expressions of restricted integer partition functions as examples.
}
\keywords{Restricted integer partition function; Quantum statistics; Symmetric
function; $S$-function; Generating function; Canonical partition function.}

\maketitle
\flushbottom


\section{Introduction}

The problem of integer partition functions is important in both statistical
mechanics and mathematics \cite{andrews1998theory,andrews2004integer}. The
restricted integer partition function counts the number of partitions which
are ways of writing an integer as a sum of other integers under certain restrictions.

\textit{The integer partition function.} In mathematics, a partition of an
integer $E$ is a way of writing $E$ as a sum of other integers
\cite{andrews1998theory}. The integer partition function of the integer$\ E$
counts the number of partitions of $E$. There are two kinds of integer
partition functions: the unrestricted integer partition function $P\left(
E\right)  $ and the restricted integer partition function $P\left(  E\text{%
$\vert$
restrictions}\right)  $. The unrestricted integer partition function $P\left(
E\right)  $ counts the number of all possible partitions
\cite{hardy1999ramanujan} and the restricted integer partition function
$P\left(  E\text{%
$\vert$
restrictions}\right)  $ counts the number of partitions under certain
restrictions \cite{andrews1998theory}.

\textit{The number of microstates in statistical mechanics.} In statistical
mechanics, the macrostate of a system can be specified completely in terms of
state variables such as the total energy $E$ and the total particle number
$N$, and a microstate is a way in which a macrostate can be realized
\cite{pathria2011statistical,reichl2009modern}. The number of microstates, or,
the state density, $\Omega\left(  E,N\right)  $ counts the number of
microstates of a given macrostate specified by the total energy $E$ and the
total particle number $N$ \cite{pathria2011statistical,reichl2009modern}. For
a non-interacting system, e.g., an ideal quantum gas, the microstate is a way
in which the total energy $E$ is distributed among the $N$ particles. That is,
the microstate is a representation of $E$ in terms of the sum of $N$
single-particle energies $\varepsilon_{i}$, where $\varepsilon_{i}$ is the
eigenvalue of the single-particle state \cite{pathria2011statistical}.
Therefore, the number of microstates of an ideal quantum gas $\Omega\left(
E,N\right)  $ counts the ways of representing $E$ as a sum of\textbf{ }$N$
single-particle energies $\varepsilon_{i}$. In statistical mechanics, various
kinds of quantum statistics are distinguished by the maximum occupation
number. The maximum occupation number is the maximum number that particles are
allowed to occupy a single-particle state. For Bose-Einstein statistics, there
is no restriction on the maximum occupation number, but for Fermi-Dirac
statistics, the maximum occupation number is $1$. Gentile statistics is a
generalization of Bose-Einstein and Fermi-Dirac statistics, whose maximum
occupation number is an arbitrary integer $q$
\cite{gentile1940itosservazioni,dai2004gentile,dai2004representation,maslov2017relationship,shen2007intermediate,maslov2017model,dai2009intermediate}%
. General statistics is a generalization of Gentile statistics, whose maximum
occupation number of different quantum states\ takes on different values
\cite{dai2009exactly}. In a word, the number of microstates $\Omega\left(
E,N\right)  $ of an ideal quantum gas is the number of representations of $E$
in terms of the sum of $N$ single-particle energies $\varepsilon_{i}$ with a
constraint that each $\varepsilon_{i}$ repeats no more than a given time, such
as $1$ for Fermi gases, $q$ for Gentile gases, $\infty$ for Bose gases, and so on.

Comparing the restricted integer partition function in mathematics and the
number of microstates in statistical mechanics, we can see that the number of
microstates of an ideal quantum gas is closely related to the restricted
integer partition function that counts the partition under the following
restrictions: (1) the number of elements (summands) is $N$, (2) the element
(summands) belongs to a set $\left\{  \varepsilon_{1},\varepsilon_{2}%
,\ldots\right\}  $, and (3) the element (summand) repeats no more than a given time.

In this paper, first, by resorting to the canonical partition function of
quantum ideal gases in statistical mechanics \cite{zhou2017canonical}, we
construct the generating function of restricted integer partition functions
corresponding to ideal Bose, Fermi, and Gentile gases, respectively. The
result shows that the generating functions for these restricted integer
partition functions are symmetric functions which are invariant under the
action of the permutation group and can be represented as linear combinations
of the $S$-function which is an important class of symmetric functions
\cite{littlewood1977theory,macdonald1998symmetric}. We also calculate the
exact expression of the restricted integer partition function from the
generating function as examples. Second, based on general statistics which is
a generalization of Gentile statistics \cite{dai2009exactly}, we introduce a
new type of restricted integer partition functions and show that a number of
restricted integer partition functions are special cases of such kind of
restricted integer partition functions.

The relation between restricted integer partition functions and statistical
mechanics has been discussed in Refs.
\cite{van1937statistical,auluck1946statistical,tran2004quantum,kubasiak2005fermi,srivatsan2006gentile,bogoliubov2007enumeration,prokhorov2012asymptotic,rovenchak2014enumeration,rovenchak2016statistical,maslov2017new,maslov2017topological}%
. The restricted integer partition function that counts partitions with
elements belonging to nature numbers corresponds to the number of microstates,
or, the state density, of the system consisting of linear simple-harmonic
oscillators in statistical mechanics
\cite{auluck1946statistical,tran2004quantum,kubasiak2005fermi,srivatsan2006gentile}%
, the restricted integer partition function that counts partitions with
distinct elements corresponds to ideal Fermi gases, the restricted integer
partition function that count partitions with elements repeating no more than
$\infty$ times corresponds to ideal Bose gases
\cite{auluck1946statistical,tran2004quantum,kubasiak2005fermi}, and the
restricted integer partition function that counts partitions with elements
repeating no more than $q$ times corresponds to ideal Gentile gases with
maximum number $q$ \cite{srivatsan2006gentile}, etc. The results given by
statistical mechanics are used to solve the integer partition function
problems and vise versa. For example, Bohr and Kalckar use the integer
partition function to calculate the density of energy levels in heavy nuclei
\cite{van1937statistical}, fluctuations in one- and three- dimensional traps
are discussed by resorting to the theory of restricted integer partition
functions \cite{grossman1999number}, the problem of integer partition
functions is addressed using the microcanonical approach in statistical
mechanics \cite{prokhorov2012asymptotic}, the quantum statistical approach is
used to estimate the expression of some restricted integer partition functions
\cite{auluck1946statistical,tran2004quantum,srivatsan2006gentile} and to
estimate the number of restricted plane partitions\textbf{ }%
\cite{rovenchak2014enumeration}, etc.

Some authors discuss the relation between the thermodynamics quantity of ideal
systems and the symmetric functions, which are both invariant under
permutations. For example, the canonical partition function for a
parastatistical system can be expressed as sums of $S$-functions\textbf{
}\cite{chaturvedi1996canonical}, the partition function of an ideal gas can be
represented in terms of symmetric functions, such as the elementary and the
complete symmetric polynomial \cite{balantekin2001partition}, the partition
function of the six vertex model is equal to a factorial Schur function which
is a generalization of $S$-functions \cite{mcnamara2009factorial}, and there
is a close correspondence between the partition function of ideal quantum
gases and certain symmetric polynomials
\cite{schmidt2002partition,gorin2015asymptotics}. Our previous work shows that
the canonical partition function for quantum ideal gases, such as ideal Bose,
Fermi, and Gentile gases, can be represented as linear combinations of the
$S$-function \cite{zhou2017canonical}. In Refs.
\cite{stanley1971theory,stanley1971theory1}, the author introduces the
symmetric function and the plane partition function. There are also
discussions on restricted integer partition functions
\cite{fel2016gaussian,andrews1998theory,mitrinovic1996handbook,andrews2004integer,richard1999enumerative,nathanson2000elementary}%
.

In this paper, based on the exact canonical partition function in statistical
mechanics, we construct the generating function for a series of restricted
integer partition functions. The result reveals a relation between the
restricted integer partition function and the symmetric function.

Especially, in this paper we introduce a type of restricted integer partition
functions based on general statistics in statistical mechanics, which is a
generalization of Bose-Einstein, Fermi-Dirac, and Gentile statistics. Many
restricted integer partition functions are special cases of this kind of
restricted integer partition functions. The restricted integer partition
function introduced in this paper enables us to consider a number of
restricted integer partition functions in a unified framework.

In section \ref{review}, a brief review of the concept of the integer
partition function and the symmetric function is given. In section \ref{gen},
the exact generating function of\ restricted integer partition functions
corresponding to ideal Bose, Fermi, and Gentile gases is given. In section
\ref{less}, we introduce a kind of restricted integer partition functions
corresponding to general statistics in statistical mechanics. In section
\ref{exam}, we calculate some expressions of restricted integer partition
functions from the generating functions as examples. Conclusions are given in
section \ref{con}. In the appendix, some expressions of generating functions
are given.

\section{The integer partition function and the symmetric function: a brief
review \label{review}}

In this section, we give a brief review on the mathematical concept of the
integer partition function and the symmetric function.

\textit{Partitions. }The partition of an integer $E$, denoted by $\left(
\lambda\right)  $, is a representation of $E$ in terms of other positive
integers which sum up to $E$. For example, the partitions for $3$ are $\left(
\lambda\right)  =\left(  3\right)  $, $\left(  \lambda\right)  ^{\prime
}=\left(  2,1\right)  $, and $\left(  \lambda\right)  ^{\prime\prime}=\left(
1^{3}\right)  $, where $1^{3}$ means $1$ is appearing $3$ times. The summand
in the partition is called the \textit{element} (or the \textit{part}),
denoted by $\lambda_{i}$. The number of elements is called the \textit{length}%
, denoted by $l_{\left(  \lambda\right)  }$. For example, the elements of the
partition $\left(  \lambda\right)  =\left(  2,1\right)  $ are $\lambda_{1}=2$
and $\lambda_{2}=1$, and the length is $l_{\left(  2,1\right)  }=2$. In a
partition $\left(  \lambda\right)  $, the elements are always arranged in
descending order: $\lambda_{1}\geqslant\lambda_{2}\geqslant\ldots
\geqslant\lambda_{l}>0$.

\textit{The unrestricted integer partition function. }The number of all
partitions is called the \textit{unrestricted integer partition function},
denoted by $P\left(  E\right)  $, e.g., $P\left(  3\right)  =3$. A famous
expression of $P\left(  E\right)  $ is given by Ramanujan
\cite{hardy1999ramanujan}. The generating function of $P\left(  E\right)  $,
$\sum_{E}P\left(  E\right)  z^{E}=%
{\displaystyle\prod\limits_{i}}
1/\left(  1-z^{i}\right)  $, is given in Ref. \cite[caput XVI]%
{euler1748introductio}.

\textit{Restricted integer partition functions. }If one counts the partitions
under certain restrictions, then the number of partitions is called the
\textit{restricted integer partition function}, denoted by $P\left(  E|\text{
restrictions }\right)  $ \cite{andrews1998theory}. For example, $P\left(
3|\text{ elements are even}\right)  =0$ since no partitions meet the
restriction and $P\left(  3|\text{ }l_{\left(  \lambda\right)  }=2\right)  =1$
since the corresponding partition is $\left(  \lambda\right)  =\left(
2,1\right)  $.

Many mathematicians, such as Gupta \cite{gupta1970partitions}, Erd\"{o}s
\cite{erdos1941distribution}, and Andrews
\cite{andrews1998theory,andrews2004integer} studied the problem of integer
partitions. Many studies devotes to the restricted integer partition function.
For example, the generating function of $P\left(  E|\text{ elements belong to
}\left\{  1,2,3\ldots,n\right\}  \text{ and }l_{\left(  \lambda\right)  }\leq
N\right)  $ is the Gauss polynomial \cite{fel2016gaussian,andrews1998theory}.
A recursive relation of $P\left(  E|\text{ elements belong to }\left\{
1,2,3\ldots,n\right\}  \right)  $ is considered in Ref.
\cite{andrews1998theory}. $P\left(  E|\text{ }l_{\left(  \lambda\right)
}=N\right)  $ and $P\left(  E|\text{ }l_{\left(  \lambda\right)  }=N\text{,
elements repeat no more than }1\text{ time}\right)  $ are considered in Refs.
\cite{mitrinovic1996handbook,andrews2004integer,richard1999enumerative}. An
approximate expression of $P\left(  E|\text{ elements belong to }\left\{
a_{1},a_{2},\ldots,a_{k}\right\}  \right)  $ is given in Ref.
\cite{nathanson2000elementary}.

\textit{Symmetric functions and }$S$\textit{-functions. }The symmetric
function $f\left(  x_{1},x_{2},\ldots,x_{n}\right)  $\textit{ }is invariant
under the action of the permutation group $S_{n}$; that is, for $\sigma\in
S_{n}$, $\sigma f\left(  x_{1},x_{2},\ldots,x_{n}\right)  =f\left(
x_{\sigma\left(  1\right)  },x_{\sigma\left(  2\right)  },\ldots
,x_{\sigma\left(  n\right)  }\right)  =f\left(  x_{1},x_{2},\ldots
,x_{n}\right)  $. The $S$-function $\left(  \lambda\right)  \left(
x_{1},x_{2},\ldots\right)  $, also known as the Schur function, is an
important kind of the symmetric function, because it forms the base of the
symmetric function space, i.e., a symmetric function can be expanded in a
linear combination of the $S$-function. The $S$-function is closely related to
the representation theory of permutation group and the unitary group
\cite{littlewood1977theory,macdonald1998symmetric}. There are also some other
symmetric functions, e.g., the elementary symmetric polynomial, the complete
symmetric polynomial, etc. \cite{macdonald1998symmetric}.

The $S$-function is defined as
\cite{littlewood1977theory,macdonald1998symmetric}
\begin{equation}
\left(  \lambda\right)  \left(  x_{1},x_{2},\ldots\right)  =\sum_{\left(
\lambda\right)  ^{\prime}}\frac{g^{\left(  \lambda\right)  }}{N!}\chi_{\left(
\lambda\right)  ^{\prime}}^{\left(  \lambda\right)  }%
{\displaystyle\prod\limits_{m=1}^{k}}
\left(  \sum_{i}x_{i}^{m}\right)  ^{a_{\left(  \lambda\right)  ^{\prime},m}},
\label{defschur}%
\end{equation}
where $N$ is the sum of elements in $\left(  \lambda\right)  $, $\sum_{\left(
\lambda\right)  ^{\prime}}$ indicates summation over all partitions $\left(
\lambda\right)  ^{\prime}$ of $N$, $g_{\left(  \lambda\right)  }$ is defined
as $g_{\left(  \lambda\right)  }=N!\left(
{\displaystyle\prod\limits_{j=1}^{N}}
j^{a_{\left(  \lambda\right)  ,j}}a_{\left(  \lambda\right)  ,j}!\right)
^{-1}$with $a_{\left(  \lambda\right)  ,m}$ counting the times of the number
$m$ appeared in $\left(  \lambda\right)  $, and $\chi_{\left(  \lambda\right)
^{\prime}}^{\left(  \lambda\right)  }$ is the simple characteristic of the
permutation group of order $N$ \cite{hamermesh1962group}.

\section{The generating function of restricted integer partition functions
counting the integer partitions with length $N$ \label{gen}}

The generating function for a restricted integer partition function $P\left(
E|\text{ restrictions}\right)  $ is defined as
\cite{andrews1998theory,gupta1970partitions}%
\begin{equation}
Z\left(  z\right)  =%
{\displaystyle\sum_{E=0}^{\infty}}
P\left(  E|\text{ restrictions}\right)  z^{E}. \label{generating}%
\end{equation}
In this section, we construct the generating functions for three restricted
integer partition functions: the restricted integer partition functions
corresponding to ideal Bose, Fermi, and Gentile gases. These three restricted
integer partition functions are important
\cite{van1937statistical,auluck1946statistical,tran2004quantum,kubasiak2005fermi,srivatsan2006gentile,bogoliubov2007enumeration,prokhorov2012asymptotic,rovenchak2014enumeration,rovenchak2016statistical,maslov2017new}%
. The generating function is an effective way to calculate the partition
function, especially in some cases that the generating function can be
obtained relatively easy. Our starting point is the canonical partition
function of these three quantum gases in statistical mechanics
\cite{zhou2017canonical}.

For a set $\left\{  \varepsilon\right\}  =\left\{  \varepsilon_{1}%
,\varepsilon_{2},\ldots\right\}  $ consisting of integers with $0<\varepsilon
_{1}<\varepsilon_{2}<\cdots$, we consider the following three restricted
integer partition functions:%

\begin{align}
P^{\left\{  \varepsilon\right\}  }\left(  E;N\right)   &  \equiv P\left(
E|\text{ elements belong to }\left\{  \varepsilon\right\}  \text{ and
}l_{\left(  \lambda\right)  }=N\right)  ,\label{0001}\\
Q^{\left\{  \varepsilon\right\}  }\left(  E;N\right)   &  \equiv P\left(
E|\text{ elements belong to }\left\{  \varepsilon\right\}  \text{, }l_{\left(
\lambda\right)  }=N\text{,}\right. \nonumber\\
&  \left.  \text{and elements repeat no more than }1\text{ time}\right)
,\label{0002}\\
P_{q}^{\left\{  \varepsilon\right\}  }\left(  E;N\right)   &  \equiv P\left(
E|\text{ elements belong to }\left\{  \varepsilon\right\}  \text{, }l_{\left(
\lambda\right)  }=N\text{,}\right. \nonumber\\
&  \left.  \text{and elements repeat no more than }q\text{ times}\right)
\label{0003}%
\end{align}
with $q=2,3,\ldots,N-1$.

In the following, for convenience we drop the superscript $\left\{
\varepsilon\right\}  $ when $\left\{  \varepsilon\right\}  =\left\{
1,2,3,4,\ldots\right\}  $. For example, $P\left(  E;N\right)  \equiv
P^{\left\{  1,2,3,4,\ldots\right\}  }\left(  E;N\right)  $, $Q\left(
E;N\right)  \equiv Q^{\left\{  1,2,3,4,\ldots\right\}  }\left(  E;N\right)  $,
and $P_{q}\left(  E;N\right)  \equiv P_{q}^{\left\{  1,2,3,4,\ldots\right\}
}\left(  E;N\right)  $.

The restricted integer partition function $P^{\left\{  \varepsilon\right\}
}\left(  E;N\right)  $, $Q^{\left\{  \varepsilon\right\}  }\left(  E;N\right)
$, and $P_{q}^{\left\{  \varepsilon\right\}  }\left(  E;N\right)  $ correspond
to ideal Bose, Fermi, and Gentile gases with the maximum occupation number
$q$, respectively
\cite{van1937statistical,auluck1946statistical,tran2004quantum,kubasiak2005fermi,srivatsan2006gentile,bogoliubov2007enumeration,prokhorov2012asymptotic,rovenchak2014enumeration,rovenchak2016statistical,maslov2017new}%
..

In the present paper, however, we construct the generating functions for
$P^{\left\{  \varepsilon\right\}  }\left(  E;N\right)  $, $Q^{\left\{
\varepsilon\right\}  }\left(  E;N\right)  $, and $P_{q}^{\left\{
\varepsilon\right\}  }\left(  E;N\right)  $. Starting from the generating
function, we calculate some expressions of $P\left(  E;N\right)  $, $Q\left(
E;N\right)  $, and $P_{q}\left(  E;N\right)  $ as examples.

\subsection{The generating function of\textbf{ }$P^{\left\{  \varepsilon
\right\}  }\left(  E;N\right)  $}

\begin{theorem}
The generating function of $P^{\left\{  \varepsilon\right\}  }\left(
E;N\right)  $ is
\begin{equation}%
{\displaystyle\sum_{E=0}^{\infty}}
P^{\left\{  \varepsilon\right\}  }\left(  E;N\right)  z^{E}=\left(  N\right)
\left(  z^{\varepsilon_{1}},z^{\varepsilon_{2}},\ldots\right)  ,
\label{muhanshu1}%
\end{equation}
where $\left(  N\right)  \left(  x_{1},x_{2},\ldots\right)  $ is the
$S$-function given by eq. (\ref{defschur}) corresponding to the partition
$\left(  \lambda\right)  =\left(  N\right)  $.
\end{theorem}

\begin{proof}
For an ideal Bose gas, the maximum occupation number is $\infty$, so the
number of microstates $\Omega_{B}\left(  E,N\right)  $ in the macrostate,
denoted by $\left(  N,E\right)  $, is the number of representations of $E$ in
terms of $N$ single-particle energies $\varepsilon_{i}$ which sum up to $E$.
By the definition of $P^{\left\{  \varepsilon\right\}  }\left(  E;N\right)  $,
eq. (\ref{0001}), one directly arrives at
\begin{equation}
P^{\left\{  \varepsilon\right\}  }\left(  E;N\right)  =\Omega_{B}\left(
E,N\right)  . \label{0005}%
\end{equation}
The canonical partition function of an ideal Bose gas is
\cite{zhou2017canonical}%
\begin{equation}
Z_{B}\left(  \beta,N\right)  =%
{\displaystyle\sum_{E=0}^{\infty}}
\Omega_{B}\left(  E,N\right)  e^{-\beta E}=\left(  N\right)  \left(
e^{-\beta\varepsilon_{1}},e^{-\beta\varepsilon_{2}},\ldots\right)  .
\label{0006}%
\end{equation}
Substituting eq. (\ref{0005}) into eq. (\ref{0006}) and equaling $z$ and
$e^{-\beta}$ give eq. (\ref{muhanshu1}).
\end{proof}

\subsection{The generating function of $Q^{\left\{  \varepsilon\right\}
}\left(  E;N\right)  $}

\begin{theorem}
The generating function of $Q^{\left\{  \varepsilon\right\}  }\left(
E;N\right)  $ is
\begin{equation}%
{\displaystyle\sum_{E=0}^{\infty}}
Q^{\left\{  \varepsilon\right\}  }\left(  E;N\right)  z^{E}=\left(
1^{N}\right)  \left(  z^{\varepsilon_{1}},z^{\varepsilon_{2}},\ldots\right)  ,
\label{muhanshu2}%
\end{equation}
where $\left(  1^{N}\right)  \left(  x_{1},x_{2},\ldots\right)  $ is the
$S$-function given by eq. (\ref{defschur}) corresponding to the partition
$\left(  \lambda\right)  =\left(  1^{N}\right)  $.
\end{theorem}

\begin{proof}
For an ideal Fermi gas, the maximum occupation number is $1$, so the number of
microstates $\Omega_{F}\left(  E,N\right)  $ in the macrostate $\left(
N,E\right)  $ is the number of representations of $E$ in terms of $N$ distinct
single-particle energies $\varepsilon_{i}$ which sum up to $E$. By the
definition of $Q^{\left\{  \varepsilon\right\}  }\left(  E;N\right)  $, eq.
(\ref{0002}), one directly arrives at
\begin{equation}
Q^{\left\{  \varepsilon\right\}  }\left(  E;N\right)  =\Omega_{F}\left(
E,N\right)  . \label{0007}%
\end{equation}
The canonical partition function of an ideal Fermi gas is
\cite{zhou2017canonical}
\begin{equation}
Z_{F}\left(  \beta,N\right)  =%
{\displaystyle\sum_{E=0}^{\infty}}
\Omega_{F}\left(  E,N\right)  e^{-\beta E}=\left(  1^{N}\right)  \left(
e^{-\beta\varepsilon_{1}},e^{-\beta\varepsilon_{2}},\ldots\right)  .
\label{0008}%
\end{equation}
Substituting eq. (\ref{0007}) into eq. (\ref{0008}) and setting $z=e^{-\beta}$
give eq. (\ref{muhanshu2}).
\end{proof}

\subsection{The generating function of\textbf{ }$P_{q}^{\left\{
\varepsilon\right\}  }\left(  E;N\right)  $}

Before going on, we first define the order of partitions
\cite{zhou2017canonical}. An integer $N$ has many partitions, we arrange the
partitions in the following order: $\left(  \lambda\right)  $, $\left(
\lambda\right)  ^{\prime}$, when $\lambda_{1}>\lambda_{1}^{\prime}$; $\left(
\lambda\right)  $, $\left(  \lambda\right)  ^{\prime}$, when $\lambda
_{1}=\lambda_{1}^{\prime}$ but $\lambda_{2}>\lambda_{2}^{\prime}$; and so on.
One keeps comparing $\lambda_{i}$ and $\lambda_{i}^{\prime}$ until all the
partitions of $N$ are arranged in the prescribed order. Let $\left(
\lambda\right)  _{J}$ denote the $J$th partition of $N$ and $\lambda_{J,i}$
denote the $i$th element in the partition $\left(  \lambda\right)  _{J}$. More
details can be found in, e.g., \cite{zhou2017canonical}.

\begin{theorem}
The generating function of $P_{q}^{\left\{  \varepsilon\right\}  }\left(
E;N\right)  $ is
\begin{equation}%
{\displaystyle\sum_{E=0}^{\infty}}
P_{q}^{\left\{  \varepsilon\right\}  }\left(  E;N\right)  z^{E}=\sum
_{I=1}^{P\left(  N\right)  \ }Q^{I}\left(  q\right)  \left(  \lambda\right)
_{I}\left(  z^{\varepsilon_{1}},z^{\varepsilon_{2}},\ldots\right)  ,
\label{muhanshu}%
\end{equation}
where $\left(  \lambda\right)  _{I}\left(  x_{1},x_{2},\ldots\right)  $ is the
$S$-function given by eq. (\ref{defschur}) corresponding to the $I$th
partition of $N$ and the coefficient
\begin{equation}
Q^{I}\left(  q\right)  =\sum_{K=1}^{P\left(  N\right)  }\left(  k_{K}%
^{I}\right)  ^{-1}\Gamma^{K}\left(  q\right)  \label{Q}%
\end{equation}
with $\Gamma^{K}\left(  q\right)  $ satisfying%
\begin{align}
\Gamma^{K}\left(  q\right)   &  =0\text{, \ \ when }\lambda_{K,1}%
>q,\nonumber\\
\Gamma^{K}\left(  q\right)   &  =1\text{, \ \ when }\lambda_{K,1}\leq q,
\label{GAMMA}%
\end{align}
and $\left(  k_{K}^{I}\right)  ^{-1}$ satisfying%
\begin{equation}
\sum_{I=1}^{P\left(  N\right)  }\left(  k_{K}^{I}\right)  ^{-1}k_{I}%
^{L}=\delta_{K}^{L} \label{KOSTKA}%
\end{equation}
with $k_{I}^{L}$ the Kostka number \cite{macdonald1998symmetric}.
\end{theorem}

\begin{proof}
For an ideal Gentile gas, the maximum occupation number is $q$\textbf{. }Thus
the number of microstates $\Omega_{q}\left(  E,N\right)  $ in the macrostate
$\left(  N,E\right)  $ is the number of representations of $E$ in terms of $N$
single-particle energies $\varepsilon_{i}$ which sum up to $E$ and each
$\varepsilon_{i}$ repeats no more than $q$ times. By the definition of
$P_{q}^{\left\{  \varepsilon\right\}  }\left(  E;N\right)  $, one directly
arrives at
\begin{equation}
P_{q}^{\left\{  \varepsilon\right\}  }\left(  E;N\right)  =\Omega_{q}\left(
E,N\right)  . \label{0009}%
\end{equation}
The canonical partition function of an ideal Gentile gas is
\cite{zhou2017canonical}
\begin{equation}
Z_{q}\left(  \beta,N\right)  =%
{\displaystyle\sum_{E=0}^{\infty}}
\Omega_{q}\left(  E,N\right)  e^{-\beta E}=\sum_{I=1}^{P\left(  N\right)
\ }Q^{I}\left(  q\right)  \left(  \lambda\right)  _{I}\left(  e^{-\beta
\varepsilon_{1}},e^{-\beta\varepsilon_{2}},\ldots\right)  . \label{00010}%
\end{equation}
Substituting eq. (\ref{0009}) into eq. (\ref{00010}) and equaling $z$ and
$e^{-\beta}$ give eq. (\ref{muhanshu}).
\end{proof}

In a word, the main result of this section is that the generating function of
the restricted integer partition function, eqs. (\ref{muhanshu1}),
(\ref{muhanshu2}), and (\ref{muhanshu}), is a symmetric function and can be
written as linear combinations of the $S$-function.

\section{The restricted integer partition function corresponding to general
statistics \label{less}}

In this section, based on general statistics in which the maximum occupation
number of a state can take on unrestricted integer values or infinity and the
maximum occupation numbers of different states may be different
\cite{dai2009exactly}, we introduce a restricted integer partition function.
General statistics is a generalization of quantum statistics, including
Bose-Einstein, Fermi-Dirac, and Gentile statistics. Bose-Einstein,
Fermi-Dirac, and Gentile statistics are special cases of general statistics,
so the restricted integer partition function corresponding to Bose-Einstein,
Fermi-Dirac, and Gentile statistics are special cases of the restricted
integer partition functions corresponding to general statistics. This means
that a number of restricted integer partition functions can be considered in a
unified framework.

\subsection{The restricted integer partition function $P_{\left\{  q\right\}
}^{\left\{  \varepsilon\right\}  }\left(  E;N\right)  $ and general
statistics}

For $\left\{  q\right\}  =\left\{  q_{1},q_{2},\ldots\right\}  $ with $q_{i}$
an integer, the restricted integer partition function $P_{\left\{  q\right\}
}^{\left\{  \varepsilon\right\}  }\left(  E;N\right)  $ is defined as
\begin{align}
P_{\left\{  q\right\}  }^{\left\{  \varepsilon\right\}  }\left(  E;N\right)
&  \equiv P\left(  E|\text{ elements belong to }\left\{  \varepsilon\right\}
\text{, }l_{\left(  \lambda\right)  }=N,\right. \nonumber\\
&  \left.  \text{and the element }\varepsilon_{i}\text{ repeats no more than
}q_{i}\text{ times}\right)  . \label{total}%
\end{align}

\begin{theorem}
The restricted integer partition function $P_{\left\{  q\right\}  }^{\left\{
\varepsilon\right\}  }\left(  E;N\right)  $ is the number of microstates
$\Omega_{\left\{  q\right\}  }\left(  E,N\right)  $ of an ideal
general-statistics gas with maximum occupation numbers $\left\{  q\right\}
=\left\{  q_{1},q_{2},\ldots\right\}  $, i.e.,
\begin{equation}
P_{\left\{  q\right\}  }^{\left\{  \varepsilon\right\}  }\left(  E;N\right)
=\Omega_{\left\{  q\right\}  }\left(  E,N\right)  . \label{ge}%
\end{equation}

\end{theorem}

\begin{proof}
For an ideal general-statistics gas, the maximum occupation number of the
state $\varepsilon_{i}$ is $q_{i}$ \cite{dai2009exactly}, so the number of
microstates $\Omega_{\left\{  q\right\}  }\left(  E,N\right)  $ in the
macrostate $\left(  N,E\right)  $ is the number of representations of $E$ in
terms of $N$ single-particle energies $\varepsilon_{i}$ which sum up to $E$
and $\varepsilon_{i}$ repeats no more than $q_{i}$ times. By the definition of
$P_{\left\{  q\right\}  }^{\left\{  \varepsilon\right\}  }\left(  E;N\right)
$, eq. (\ref{total}), one directly arrives at eq. (\ref{ge}).
\end{proof}

\subsection{The two-variable generating function of $P_{\left\{  q\right\}
}^{\left\{  \varepsilon\right\}  }\left(  E;N\right)  $}

In this section, we give the two-variable generating function of $P_{\left\{
q\right\}  }^{\left\{  \varepsilon\right\}  }\left(  E;N\right)  $.

Some kinds of restricted integer partition functions need to be described by
two-variable generating functions \cite{andrews1998theory}:
\begin{equation}
\Xi\left(  z,x\right)  =\sum_{N,E=0}^{\infty}P\left(  E|\text{ restrictions on
}l_{\left(  \lambda\right)  }\text{ and other restrictions}\right)  z^{E}%
x^{N}, \label{two}%
\end{equation}
where $N$ is the length $l_{\left(  \lambda\right)  }$ of the partition.

The two-variable generating function of the restricted integer partition
function $P_{\left\{  q\right\}  }^{\left\{  \varepsilon\right\}  }\left(
E;N\right)  $ defined by eq. (\ref{total}) is then%
\begin{equation}
\Xi\left(  z,x\right)  =\sum_{N,E=0}^{\infty}P_{\left\{  q\right\}
}^{\left\{  \varepsilon\right\}  }\left(  E;N\right)  z^{E}x^{N}.
\label{gener}%
\end{equation}

\begin{theorem}
The two-variable generating function of the restricted integer partition
function $P_{\left\{  q\right\}  }^{\left\{  \varepsilon\right\}  }\left(
E;N\right)  $ is
\begin{equation}
\sum_{N,E=0}^{\infty}P_{\left\{  q\right\}  }^{\left\{  \varepsilon\right\}
}\left(  E;N\right)  z^{E}x^{N}=%
{\displaystyle\prod\limits_{i}}
\frac{1-\left(  z^{\varepsilon_{i}}x\right)  ^{q_{i}+1}}{1-z^{\varepsilon_{i}%
}x}. \label{gge}%
\end{equation}

\end{theorem}

\begin{proof}
The grand canonical partition function of a general-statistics gas is
\cite{dai2009exactly}%
\begin{equation}
\Xi_{\left\{  q\right\}  }\left(  \alpha,\beta\right)  =\sum_{E,N=0}^{\infty
}\Omega_{\left\{  q\right\}  }e^{-\beta E}e^{-\alpha N}=%
{\displaystyle\prod\limits_{i}}
\frac{1-\left[  \left(  e^{-\beta}\right)  ^{\varepsilon_{i}}e^{-\alpha
}\right]  ^{q_{i}+1}}{1-\left(  e^{-\beta}\right)  ^{\varepsilon_{i}%
}e^{-\alpha}}. \label{002}%
\end{equation}
Substituting eq. (\ref{ge}) into eq. (\ref{002}) and setting $z=e^{-\beta}$
and $x=e^{-\alpha}$ give eq. (\ref{gge}).
\end{proof}

\subsection{Restricted integer partition functions as special cases of
$P_{\left\{  q\right\}  }^{\left\{  \varepsilon\right\}  }\left(  E;N\right)
$}

A series of restricted integer partition functions are special cases of the
restricted integer partition function $P_{\left\{  q\right\}  }^{\left\{
\varepsilon\right\}  }\left(  E;N\right)  $ defined by eq. (\ref{total}). We
show that starting from the restricted integer partition function $P_{\left\{
q\right\}  }^{\left\{  \varepsilon\right\}  }\left(  E;N\right)  $, one can
obtain many restricted integer partition functions.

In the following, we consider three types of restricted integer partition
functions: the restricted integer partition functions counting partitions with
length $l_{\left(  \lambda\right)  }=N$, with length $l_{\left(
\lambda\right)  }\leq N$, and without restrictions on the length $l_{\left(
\lambda\right)  }$.

\subsubsection{Restricted integer partition functions counting partitions with
length $l_{\left(  \lambda\right)  }=N$}

In this section, we consider the following three restricted integer partition
functions: $P^{\left\{  \varepsilon\right\}  }\left(  E;N\right)  $, defined
in eq. (\ref{0001}) corresponding to Bose-Einstein statistics, $Q^{\left\{
\varepsilon\right\}  }\left(  E;N\right)  $, defined in eq. (\ref{0002})
corresponding to Fermi-Dirac statistics, and $P_{q}^{\left\{  \varepsilon
\right\}  }\left(  E;N\right)  $, defined in eq. (\ref{0003}) corresponding to
Gentile statistics.

These three restricted integer partition functions, $P^{\left\{
\varepsilon\right\}  }\left(  E;N\right)  $, $Q^{\left\{  \varepsilon\right\}
}\left(  E;N\right)  $, and $P_{q}^{\left\{  \varepsilon\right\}  }\left(
E;N\right)  $, are three special cases of the restricted integer partition
function corresponding to general statistics $P_{\left\{  q\right\}
}^{\left\{  \varepsilon\right\}  }\left(  E;N\right)  $ defined by eq.
(\ref{total}):
\begin{align}
P^{\left\{  \varepsilon\right\}  }\left(  E;N\right)   &  =P_{\left\{
q_{i}=\infty\right\}  }^{\left\{  \varepsilon\right\}  }\left(  E;N\right)
,\label{a00}\\
Q^{\left\{  \varepsilon\right\}  }\left(  E;N\right)   &  =P_{\left\{
q_{i}=1\right\}  }^{\left\{  \varepsilon\right\}  }\left(  E;N\right)
,\label{a333}\\
P_{q}^{\left\{  \varepsilon\right\}  }\left(  E;N\right)   &  =P_{\left\{
q_{i}=q\right\}  }^{\left\{  \varepsilon\right\}  }\left(  E;N\right)  .
\label{aaa1}%
\end{align}
One can see from eqs. (\ref{a00}), (\ref{a333}), and (\ref{aaa1}) that the
restricted integer partition functions $P^{\left\{  \varepsilon\right\}
}\left(  E;N\right)  $, $Q^{\left\{  \varepsilon\right\}  }\left(  E;N\right)
$, and $P_{q}^{\left\{  \varepsilon\right\}  }\left(  E;N\right)  $ can be
obtained by setting $q_{i}=\infty$, $q_{i}=1$, and $q_{i}=q$ in the restricted
integer partition function $P_{\left\{  q\right\}  }^{\left\{  \varepsilon
\right\}  }\left(  E;N\right)  $ given by eq. (\ref{total}). The two-variable
generating functions of these three restricted integer partition functions can
be then obtained by setting $q_{i}=\infty$, $q_{i}=1$, and $q_{i}=q$ in the
two-variable generating functions (\ref{gge}), respectively: $\sum
_{E,N}P^{\left\{  \varepsilon\right\}  }\left(  E;N\right)  z^{E}x^{N}=%
{\displaystyle\prod\limits_{i=1}}
1/\left(  1-z^{\varepsilon_{i}}x\right)  $, $\sum_{E,N}Q^{\left\{
\varepsilon\right\}  }\left(  E;N\right)  z^{E}x^{N}=%
{\displaystyle\prod\limits_{i=1}}
\left(  1+z^{\varepsilon_{i}}x\right)  $, and $\sum_{E,N}P_{q}^{\left\{
\varepsilon\right\}  }\left(  E;N\right)  z^{E}x^{N}=%
{\displaystyle\prod\limits_{i=1}}
\left[  1-\left(  z^{\varepsilon_{i}}x\right)  ^{q+1}\right]  /\left(
1-z^{\varepsilon_{i}}x\right)  $. These results agree with the result in Ref.
\cite{andrews2004integer}. In Ref. \cite{andrews2004integer}, two-variable
generating functions of $P^{\left\{  \varepsilon\right\}  }\left(  E;N\right)
$, $Q^{\left\{  \varepsilon\right\}  }\left(  E;N\right)  $, and
$P_{q}^{\left\{  \varepsilon\right\}  }\left(  E;N\right)  $ are obtained by
recognizing that the exponent in the expansion of the two-variable generating
of $x$ will keep count of how many parts are used in each partition.

\subsubsection{Restricted integer partition functions counting partitions with
length $l_{\left(  \lambda\right)  }\leq N$}

The restricted integer partition function counting partitions with length
$l_{\left(  \lambda\right)  }\leq N$ is a kind of important restricted integer
partition functions \cite{andrews1998theory,andrews2004integer}.\textbf{ }In
this section, we consider the following restricted integer partition function:%
\begin{align}
p_{\left\{  q\right\}  }^{\left\{  \varepsilon\right\}  }\left(  E;N\right)
&  \equiv P\left(  E|\text{ elements belong to }\left\{  \varepsilon\right\}
\text{, }l_{\left(  \lambda\right)  }\leq N\text{, }\right. \nonumber\\
&  \left.  \text{and the element }\varepsilon_{i}\text{ repeats no more than
}q_{i}\text{ times}\right)  . \label{total2}%
\end{align}
There are three important special cases\ of the restricted integer partition
function $p_{\left\{  q\right\}  }^{\left\{  \varepsilon\right\}  }\left(
E;N\right)  $:%
\begin{align}
p^{\left\{  \varepsilon\right\}  }\left(  E;N\right)   &  \equiv P\left(
E|\text{ elements belong to }\left\{  \varepsilon\right\}  \text{ and
}l_{\left(  \lambda\right)  }\leq N\right)  ,\label{10001}\\
q^{\left\{  \varepsilon\right\}  }\left(  E;N\right)   &  \equiv P\left(
E|\text{ elements belong to }\left\{  \varepsilon\right\}  \text{, }l_{\left(
\lambda\right)  }\leq N\text{,}\right. \nonumber\\
&  \left.  \text{and elements repeat no more than }1\text{ time}\right)
,\label{10002}\\
p_{q}^{\left\{  \varepsilon\right\}  }\left(  E;N\right)   &  \equiv P\left(
E|\text{ elements belong to }\left\{  \varepsilon\right\}  \text{, }l_{\left(
\lambda\right)  }\leq N\text{,}\right. \nonumber\\
&  \left.  \text{and elements repeat no more than }q\text{ times}\right)  .
\label{10003}%
\end{align}

The restricted integer partition function $p_{\left\{  q\right\}  }^{\left\{
\varepsilon\right\}  }\left(  E;N\right)  $ is a special case of the
restricted integer partition function corresponding to general statistics
$P_{\left\{  q\right\}  }^{\left\{  \varepsilon\right\}  }\left(  E;N\right)
$ defined by eq. (\ref{total}):
\begin{equation}
p_{\left\{  q\right\}  }^{\left\{  \varepsilon\right\}  }\left(  E;N\right)
=P_{\left\{  \infty,q\right\}  }^{\left\{  0,\varepsilon\right\}  }\left(
E;N\right)  , \label{3001}%
\end{equation}
where $\left\{  0,\varepsilon\right\}  $ denotes $\left\{  0,\varepsilon
_{1},\varepsilon_{2},\ldots\right\}  $ and $\left\{  \infty,q\right\}  $
denotes $\left\{  \infty,q_{1},q_{2},\ldots\right\}  $. Then the three special
cases of $p_{\left\{  q\right\}  }^{\left\{  \varepsilon\right\}  }\left(
E;N\right)  $ are
\begin{align}
p^{\left\{  \varepsilon\right\}  }\left(  E;N\right)   &  =P^{\left\{
0,\varepsilon\right\}  }\left(  E;N\right)  ,\label{100004}\\
q^{\left\{  \varepsilon\right\}  }\left(  E;N\right)   &  =P_{\left\{
\infty,1,1,\ldots\right\}  }^{\left\{  0,\varepsilon_{1},\varepsilon
_{2},\ldots\right\}  }\left(  E;N\right)  ,\label{100009}\\
p_{q}^{\left\{  \varepsilon\right\}  }\left(  E;N\right)   &  =P_{\left\{
\infty,q,q,\ldots\right\}  }^{\left\{  0,\varepsilon_{1},\varepsilon
_{2},\ldots\right\}  }\left(  E;N\right)  . \label{100012}%
\end{align}

The two-variable generating function of the restricted integer partition
function $p_{\left\{  q\right\}  }^{\left\{  \varepsilon\right\}  }\left(
E;N\right)  $ defined by eq. (\ref{total2}) can be obtained by substituting
eq. (\ref{3001}) into eq. (\ref{gge}):%
\begin{equation}
\sum_{E,N=0}p_{\left\{  q\right\}  }^{\left\{  \varepsilon\right\}  }\left(
E;N\right)  z^{E}x^{N}=\frac{1}{1-x}%
{\displaystyle\prod\limits_{i=1}}
\frac{1-\left(  z^{\varepsilon_{i}}x\right)  ^{q_{i}+1}}{1-z^{\varepsilon_{i}%
}x}. \label{r3001}%
\end{equation}

The two-variable generating functions of the three special cases of the
restricted integer partition function $p_{\left\{  q\right\}  }^{\left\{
\varepsilon\right\}  }\left(  E;N\right)  $ are then%
\begin{align}
\sum_{N,E=0}^{\infty}p^{\left\{  \varepsilon\right\}  }\left(  E;N\right)
z^{E}x^{N}  &  =\frac{1}{1-x}%
{\displaystyle\prod\limits_{i}}
\frac{1}{1-z^{\varepsilon_{i}}x},\label{100007}\\
\sum_{N,E=0}^{\infty}q^{\left\{  \varepsilon\right\}  }\left(  E;N\right)
z^{E}x^{N}  &  =\frac{1}{1-x}%
{\displaystyle\prod\limits_{i}}
\left(  1+z^{\varepsilon_{i}}x\right)  ,\label{100010}\\
\sum_{N,E=0}^{\infty}p_{q}^{\left\{  \varepsilon\right\}  }\left(  E;N\right)
z^{E}x^{N}  &  =\frac{1}{1-x}%
{\displaystyle\prod\limits_{i=1}}
\frac{1-\left(  z^{\varepsilon_{i}}x\right)  ^{q+1}}{1-z^{\varepsilon_{i}}x}.
\label{100013}%
\end{align}

Additionally, besides the two-variable generating function given by eq.
(\ref{100007}), we can also obtain a one-variable generating function of the
restricted partition integer function $p^{\left\{  \varepsilon\right\}
}\left(  E;N\right)  $:%
\begin{equation}
\sum_{E}p^{\left\{  \varepsilon\right\}  }\left(  E;N\right)  z^{E}=\left(
N\right)  \left(  1,z^{\varepsilon_{1}},z^{\varepsilon_{2}},\ldots\right)  .
\label{100005}%
\end{equation}
This can be achieved by substituting eq. (\ref{100004}) into eq.
(\ref{muhanshu1}).

\subsubsection{Restricted integer partition functions counting partitions
without restrictions on the length $l_{\left(  \lambda\right)  }$}

The restricted integer partition function counting partitions without
restrictions on $l_{\left(  \lambda\right)  }$ is another kind of important
restricted integer partition functions
\cite{andrews1998theory,andrews2004integer}. In this section, we consider the
following restricted integer partition function:%
\begin{align}
P_{\left\{  q\right\}  }^{\left\{  \varepsilon\right\}  }\left(  E\right)   &
\equiv P\left(  E|\text{ elements belong to }\left\{  \varepsilon\right\}
\right. \nonumber\\
&  \left.  \text{and the element }\varepsilon_{i}\text{ repeats no more than
}q_{i}\text{ times}\right)  . \label{total3}%
\end{align}
There are three important special cases\ of the restricted integer partition
function $P_{\left\{  q\right\}  }^{\left\{  \varepsilon\right\}  }\left(
E\right)  $:%
\begin{align}
P^{\left\{  \varepsilon\right\}  }\left(  E\right)   &  \equiv P\left(
E|\text{ elements belong to }\left\{  \varepsilon\right\}  \right)
,\label{20001}\\
Q^{\left\{  \varepsilon\right\}  }\left(  E\right)   &  \equiv P\left(
E|\text{ elements belong to }\left\{  \varepsilon\right\}  \right. \nonumber\\
&  \left.  \text{and elements repeat no more than }1\text{ time}\right)
,\label{20002}\\
P_{q}^{\left\{  \varepsilon\right\}  }\left(  E\right)   &  \equiv P\left(
E|\text{ elements belong to }\left\{  \varepsilon\right\}  \right. \nonumber\\
&  \left.  \text{and elements repeat no more than }q\text{ times}\right)  .
\label{20003}%
\end{align}

The restricted integer partition function can be represented by $P_{\left\{
q\right\}  }^{\left\{  \varepsilon\right\}  }\left(  E;N\right)  $ as%
\begin{equation}
P_{\left\{  q\right\}  }^{\left\{  \varepsilon\right\}  }\left(  E\right)
=\sum_{N}P_{\left\{  q\right\}  }^{\left\{  \varepsilon\right\}  }\left(
E;N\right)  . \label{peandpen}%
\end{equation}

The three special cases of $P_{\left\{  q\right\}  }^{\left\{  \varepsilon
\right\}  }\left(  E\right)  $ then can be represented as%
\begin{align}
P^{\left\{  \varepsilon\right\}  }\left(  E\right)   &  =\sum_{N}P^{\left\{
\varepsilon\right\}  }\left(  E;N\right)  ,\label{4003}\\
Q^{\left\{  \varepsilon\right\}  }\left(  E\right)   &  =\sum_{N}Q^{\left\{
\varepsilon\right\}  }\left(  E;N\right)  ,\label{4006}\\
P_{q}^{\left\{  \varepsilon\right\}  }\left(  E\right)   &  =\sum_{N}%
P_{q}^{\left\{  \varepsilon\right\}  }\left(  E;N\right)  . \label{4009}%
\end{align}

The generating function of the restricted integer partition function
$P_{\left\{  q\right\}  }^{\left\{  \varepsilon\right\}  }\left(  E\right)  $
can be obtained by setting $x=1$ in eq. (\ref{gge}):%
\begin{equation}%
{\displaystyle\sum_{E=0}^{\infty}}
P_{\left\{  q\right\}  }^{\left\{  \varepsilon\right\}  }\left(  E\right)
z^{E}=%
{\displaystyle\prod\limits_{i}}
\frac{1-\left(  z^{\varepsilon_{i}}\right)  ^{q_{i}+1}}{1-z^{\varepsilon_{i}}%
}. \label{4001}%
\end{equation}

Moreover, by setting $q_{i}=\infty$, $q_{i}=1$, and $q_{i}=q$ in eq.
(\ref{4001}), we can obtain the generating functions of restricted integer
partition function $P^{\left\{  \varepsilon\right\}  }\left(  E\right)  $,
$Q^{\left\{  \varepsilon\right\}  }\left(  E\right)  $, and $P_{q}^{\left\{
\varepsilon\right\}  }\left(  E\right)  $: $\sum_{E}P^{\left\{  \varepsilon
\right\}  }\left(  E\right)  z^{E}=%
{\displaystyle\prod\limits_{i}}
1/\left(  1-z^{\varepsilon_{i}}\right)  $, $\sum_{E}Q^{\left\{  \varepsilon
\right\}  }\left(  E\right)  z^{E}=%
{\displaystyle\prod\limits_{i}}
\left(  1+z^{\varepsilon_{i}}\right)  $, and $\sum_{E=0}^{\infty}%
P_{q}^{\left\{  \varepsilon\right\}  }\left(  E\right)  z^{E}=%
{\displaystyle\prod\limits_{i}}
\left[  1-\left(  z^{\varepsilon_{i}}\right)  ^{q+1}\right]  /\left(
1-z^{\varepsilon_{i}}\right)  $. The generating functions of $P^{\left\{
\varepsilon\right\}  }\left(  E\right)  $, $Q^{\left\{  \varepsilon\right\}
}\left(  E\right)  $, and $P_{q}^{\left\{  \varepsilon\right\}  }\left(
E\right)  $ agree with the result in Refs.
\cite{andrews1998theory,andrews2004integer} in which the generating functions
is obtained by recognizing that the exponent in the expansion of the
generating function will count the number of partitions.

\subsubsection{A relation between the $S$-function and the Gauss polynomial:
the restricted integer partition function $p^{\left\{  1,2,3,\ldots,n\right\}
}\left(  E;N\right)  $}

In this section, we give a relation between the $S$-function and the Gauss
polynomial by inspection of the restricted integer partition function
$p^{\left\{  1,2,3,\ldots,n\right\}  }\left(  E;N\right)  =P\left(  E|\text{
elements belong to }\left\{  1,2,3,\ldots,n\right\}  \text{ and }l_{\left(
\lambda\right)  }\leq N\right)  $, a special case of the restricted integer
partition function $p^{\left\{  \varepsilon\right\}  }\left(  E;N\right)  $
defined in eq. (\ref{0001}) with $\left\{  \varepsilon\right\}  =\left\{
1,2,3,\ldots,n\right\}  $.

\begin{theorem}
A relation between the $S$-function $\left(  N\right)  \left(  1,z,z^{2}%
,\ldots,z^{n}\right)  $ and the Gauss polynomial $G\left(  n,N;x\right)  =%
{\displaystyle\prod\limits_{i=1}^{N+n}}
\left(  1-x^{i}\right)  /\left[
{\displaystyle\prod\limits_{\mu=1}^{n}}
\left(  1-x^{\mu}\right)
{\displaystyle\prod\limits_{v=1}^{N}}
\left(  1-x^{v}\right)  \right]  $ is
\begin{equation}
\left(  N\right)  \left(  1,z,z^{2},\ldots,z^{n}\right)  =G\left(
n,N;z\right)  . \label{Co}%
\end{equation}

\end{theorem}

\begin{proof}
The generating function of $p^{\left\{  1,2,3,\ldots,n\right\}  }\left(
E;N\right)  $, eq. (\ref{100005}) with $\left\{  \varepsilon\right\}
=\left\{  1,2,3,\ldots,n\right\}  $, can be written as
\begin{equation}
\sum_{E}p^{\left\{  1,2,3,\ldots,n\right\}  }\left(  E;N\right)  z^{E}=\left(
N\right)  \left(  1,z,z^{2},\ldots,z^{n}\right)  , \label{GLESS}%
\end{equation}
i.e., the generating function of $p^{\left\{  1,2,3,\ldots,n\right\}  }\left(
E;N\right)  $ is the $S$-function $\left(  N\right)  \left(  1,z,z^{2}%
,\ldots,z^{n}\right)  $. It is also shown in Refs.
\cite{fel2016gaussian,andrews1998theory} that the generating function of
$p^{\left\{  1,2,3,\ldots,n\right\}  }\left(  E;N\right)  $ is a Gauss
polynomial:%
\begin{equation}
\sum_{E=0}^{mn}p^{\left\{  1,2,3,\ldots,n\right\}  }\left(  E,m\right)
z^{E}=G\left(  n,N;z\right)  . \label{GLESS1}%
\end{equation}

Comparing eqs. (\ref{GLESS}) and (\ref{GLESS1}), we arrive at eq. (\ref{Co}) directly.
\end{proof}

\subsubsection{The generating function of restricted integer partition
functions $P^{\left\{  0,1,2,\ldots\right\}  }\left(  E;N\right)  $,
$Q^{\left\{  0,1,2,\ldots\right\}  }\left(  E;N\right)  $, $P^{\left\{
1^{2},2^{2},\ldots\right\}  }\left(  E;N\right)  $, and $Q^{\left\{
1^{2},2^{2},\ldots\right\}  }\left(  E;N\right)  $ and quantum gases}

In this section, we give the generating function of restricted integer
partition functions on the nature number. The restricted integer partition
function is closely related to a one-dimensional ideal quantum gas in an
external harmonic-oscillator potential.

The restricted integer partition function $P^{\left\{  0,1,2,\ldots\right\}
}\left(  E;N\right)  $ and $Q^{\left\{  0,1,2,\ldots\right\}  }\left(
E;N\right)  $ corresponding to a one-dimensional ideal quantum gas in an
external harmonic-oscillator potential with ground-state\ energy equaling to
$0$, while the restricted integer partition function $P\left(  E;N\right)  $
and $Q\left(  E;N\right)  $ corresponding to one-dimensional ideal quantum
gases in an external harmonic-oscillator potential with nonzero
ground-state\ energy
\cite{auluck1946statistical,tran2004quantum,kubasiak2005fermi,srivatsan2006gentile}%
.

Setting $\left\{  \varepsilon\right\}  =\left\{  0,1,2,\ldots\right\}  $ and
using eqs. (\ref{muhanshu1}) and (\ref{muhanshu2}) give the generating
function of $P^{\left\{  0,1,2,\ldots\right\}  }\left(  E;N\right)  $ and
$Q^{\left\{  0,1,2,\ldots\right\}  }\left(  E;N\right)  $:
\begin{align}
\sum_{E=0}^{\infty}P^{\left\{  0,1,2,\ldots\right\}  }\left(  E;N\right)
z^{E}  &  =%
{\displaystyle\prod_{i=1}^{N}}
\frac{1}{\left(  1-z^{i}\right)  },\label{al10}\\
\sum_{E=0}^{\infty}Q^{\left\{  0,1,2,\ldots\right\}  }\left(  E;N\right)
z^{E}  &  =z^{N(N-1)/2}%
{\displaystyle\prod_{i=1}^{N}}
\frac{1}{\left(  1-z^{i}\right)  }. \label{al20}%
\end{align}
We can also obtain the generating function of the restricted integer partition
function $P\left(  E;N\right)  $ and $Q\left(  E;N\right)  $. Setting
$\left\{  \varepsilon\right\}  =\left\{  1,2,\ldots\right\}  $ and using eqs.
(\ref{muhanshu1}) and (\ref{muhanshu2}) give $\sum_{E}P\left(  E;N\right)
z^{E}=z^{N}%
{\displaystyle\prod_{i=1}^{N}}
1/\left(  1-z^{i}\right)  $ and $\sum_{E}Q\left(  E;N\right)  z^{E}%
=z^{N(N+1)/2}%
{\displaystyle\prod_{i=1}^{N}}
1/\left(  1-z^{i}\right)  $. The generating functions of the restricted
integer partition function $P\left(  E;N\right)  $ and $Q\left(  E;N\right)  $
agree with the result in Ref. \cite{gupta1970partitions}. In Ref.
\cite{gupta1970partitions}, the generating functions of the restricted integer
partition function $P\left(  E;N\right)  $ and $Q\left(  E;N\right)  $ are
obtained by selecting the coefficient of $z^{E}x^{N}$ in the expansion of $%
{\displaystyle\prod\limits_{i=1}}
1/\left(  1-z^{i}x\right)  $ and $%
{\displaystyle\prod\limits_{i=1}}
\left(  1+z^{i}x\right)  $,\ respectively.

Now we consider the restricted integer partition function on square numbers.

The restricted integer partition functions $P^{\left\{  1^{2},2^{2}%
,\ldots\right\}  }\left(  E;N\right)  $ and $Q^{\left\{  1^{2},2^{2}%
,\ldots\right\}  }\left(  E;N\right)  $ are actually the state density of an
ideal quantum gas in a one-dimensional periodic box: $P^{\left\{  1^{2}%
,2^{2},\ldots\right\}  }\left(  E;N\right)  $ for ideal Bose gases and
$Q^{\left\{  1^{2},2^{2},\ldots\right\}  }\left(  E;N\right)  $ for ideal
Fermi gases. We can calculate the exact state density from the exact
generating function of $P^{\left\{  1^{2},2^{2},\ldots\right\}  }\left(
E;N\right)  $ and $Q^{\left\{  1^{2},2^{2},\ldots\right\}  }\left(
E;N\right)  $, eqs. (\ref{muhanshu1}) and (\ref{muhanshu2}) with $\left\{
\varepsilon\right\}  =\left\{  1^{2},2^{2},\ldots\right\}  $.

For $N=2$, the restricted integer partition functions $P^{\left\{  1^{2}%
,2^{2},\ldots\right\}  }\left(  E;2\right)  $ and $Q^{\left\{  1^{2}%
,2^{2},\ldots\right\}  }\left(  E;2\right)  $ are the state density of two
quantum particles.

Setting $\left\{  \varepsilon\right\}  =\left\{  1^{2},2^{2},\ldots\right\}  $
and using eqs. (\ref{muhanshu1}) and (\ref{muhanshu2}) with $N=2$ give the
exact generating function of $P^{\left\{  1^{2},2^{2},\ldots\right\}  }\left(
E;N\right)  $ and $Q^{\left\{  1^{2},2^{2},\ldots\right\}  }\left(
E;N\right)  $,
\begin{equation}
\sum_{E=0}^{\infty}P^{\left\{  1^{2},2^{2},\ldots\right\}  }\left(
E;2\right)  z^{E}=\frac{1}{8}\theta_{3}^{2}\left(  0,z\right)  -\frac{1}%
{2}\theta_{3}\left(  0,z\right)  +\frac{3}{8}, \label{b2}%
\end{equation}
where $\theta_{3}\left(  u,q\right)  $ is the elliptic theta function
\cite{akhiezerelements}, and%
\begin{equation}
\sum_{E=0}^{\infty}Q^{\left\{  1^{2},2^{2},\ldots\right\}  }\left(
E;2\right)  z^{E}=\frac{1}{8}\theta_{3}^{2}\left(  0,z\right)  -\frac{1}{8}.
\label{f2}%
\end{equation}

\section{Calculating $P^{\left\{  \varepsilon\right\}  }\left(  E;N\right)  $,
$Q^{\left\{  \varepsilon\right\}  }\left(  E;N\right)  $, and $P_{q}^{\left\{
\varepsilon\right\}  }\left(  E;N\right)  $ from the generating functions:
examples \label{exam}}

The generating function of the restricted integer partition function, eqs.
(\ref{generating}), is indeed a result of the $Z$-transform performed on the
restricted integer partition function \cite{steuding2007value}. Thus, we can
obtain the restricted integer partition function by applying an inverse
$Z$-transform on the generating function \cite{andrews1998theory}:%
\begin{equation}
P\left(  E|\text{ conditions }\right)  =\frac{1}{2\pi i}%
{\displaystyle\oint\limits_{C}}
Z\left(  z\right)  z^{-n-1}dz, \label{inverseZ}%
\end{equation}
where $Z\left(  z\right)  $ is the generating function defined in eq.
(\ref{generating}).

The exact generating functions of $P^{\left\{  \varepsilon\right\}  }\left(
E;N\right)  $, $Q^{\left\{  \varepsilon\right\}  }\left(  E;N\right)  $, and
$P_{q}^{\left\{  \varepsilon\right\}  }\left(  E;N\right)  $ are given in eqs.
(\ref{muhanshu1}), (\ref{muhanshu2}), and (\ref{muhanshu}), respectively. One
can obtain an explicit expression of the generating function once $N$ and the
set $\left\{  \varepsilon\right\}  $ are determined. For example, the
generating function for restricted integer partition functions on nature
numbers and square numbers can be obtained by setting $\left\{  \varepsilon
\right\}  =\left\{  1,2,3,\ldots\right\}  $ and $\left\{  \varepsilon\right\}
=\left\{  1^{2},2^{2},3^{2},\ldots\right\}  $, respectively. Therefore, by
substituting the generating functions (\ref{muhanshu1}), (\ref{muhanshu2}),
and (\ref{muhanshu}) into eq. (\ref{inverseZ}), one can in principle obtain
the exact expression of the restricted integer partition functions
$P^{\left\{  \varepsilon\right\}  }\left(  E;N\right)  $, $Q^{\left\{
\varepsilon\right\}  }\left(  E;N\right)  $, and $P_{q}^{\left\{
\varepsilon\right\}  }\left(  E;N\right)  $.

\subsection{Expressions of $P^{\left\{  \varepsilon\right\}  }\left(
E;N\right)  $, $Q\left(  E;N\right)  $, and $P_{q}\left(  E;N\right)  $ for
$N=2$, $3$, $4$, and $5$}

In this section, using eq. (\ref{inverseZ}), we give the exact expressions of
$P^{\left\{  \varepsilon\right\}  }\left(  E;N\right)  $, $Q\left(
E;N\right)  $, and $P_{q}\left(  E;N\right)  $ for $N=2$, $3$, $4$, and $5$ as
examples. Moreover, in the appendix, we list the expressions of the generating
functions of $P^{\left\{  \varepsilon\right\}  }\left(  E;N\right)  $,
$Q^{\left\{  \varepsilon\right\}  }\left(  E;N\right)  $, and $P_{q}^{\left\{
\varepsilon\right\}  }\left(  E;N\right)  $ for $N=2$, $3$, $4$, $5$, and $6$.
The expressions of the generating functions of $P\left(  E;N\right)  $,
$Q\left(  E;N\right)  $, and $P_{q}\left(  E;N\right)  $ can be obtained by
setting $\left\{  \varepsilon\right\}  =\left\{  1,2,3,\ldots\right\}  $ in
the expression of generating functions listed in Appendix \ref{QIq}.

$N=2$. The generating function of $Q\left(  E;2\right)  $\ by eq.
(\ref{muhanshu2}) reads
\begin{equation}
\sum_{E}Q\left(  E;2\right)  z^{E}=\frac{z^{3}}{\left(  z-1\right)
^{2}\left(  1+z\right)  }. \label{21}%
\end{equation}
Substituting eq. (\ref{21})\ into eq. (\ref{inverseZ}) gives%
\begin{equation}
Q\left(  E;2\right)  =\frac{E}{2}-\frac{1}{4}\left(  -1\right)  ^{E}-\frac
{3}{4}.
\end{equation}

$N=3$. The generating function for $Q\left(  E;3\right)  $ by eq.
(\ref{muhanshu2}) reads
\begin{equation}
\sum_{E}Q\left(  E;3\right)  z^{E}=-\frac{z^{6}}{\left(  z-1\right)
^{3}\left(  1+z\right)  \left(  1+z+z^{2}\right)  }. \label{31}%
\end{equation}
Substituting eq. (\ref{31})\ into eq. (\ref{inverseZ}) gives%
\begin{equation}
Q\left(  E;3\right)  =\frac{1}{72}\left[  6E^{2}-36E+16\cos\left(  \frac
{2E\pi}{3}\right)  +9\left(  -1\right)  ^{E}+47\right]  .
\end{equation}
The generating function of $P_{2}\left(  E;3\right)  $ by eq. (\ref{muhanshu})
reads%
\begin{equation}
\sum_{E}P_{2}\left(  E;3\right)  z^{E}=\frac{\left[  z\left(  z-1\right)
-1\right]  z^{4}}{\left(  z-1\right)  ^{3}\left(  z+1\right)  \left(
1+z+z^{2}\right)  }. \label{32}%
\end{equation}
Substituting eq. (\ref{32})\ into eq. (\ref{inverseZ}) gives%
\begin{equation}
P_{2}\left(  E;3\right)  =\frac{1}{72}\left[  6E^{2}-32\cos\left(  \frac{2}%
{3}E\pi\right)  -9\left(  -1\right)  ^{E}-31\right]  .
\end{equation}

$N=4$. The generating function of $Q\left(  E;4\right)  $ by eq.
(\ref{muhanshu2}) reads%
\begin{equation}
\sum_{E}Q\left(  E;4\right)  z^{E}=\frac{z^{10}}{\left(  z-1\right)
^{4}\left(  1+z\right)  ^{2}\left(  1+z^{2}\right)  \left(  1+z+z^{2}\right)
}. \label{41}%
\end{equation}
Substituting eq. (\ref{41})\ into eq. (\ref{inverseZ}) gives%
\begin{align}
Q\left(  E;4\right)   &  =\frac{1}{288}\left[  2E^{3}-30E^{2}+135E+9E\left(
-1\right)  ^{E}-45\left(  -1\right)  ^{E}\right. \nonumber\\
&  \left.  -18\left(  -i\right)  ^{E}-18i^{E}-32U_{E}\left(  -\frac{1}%
{2}\right)  -175\right]  , \label{111}%
\end{align}
where $U_{E}\left(  x\right)  $ is the Chebyshev polynomials of second kind
\cite{dette1995note}.\textbf{ }The generating function of $P_{2}\left(
E;4\right)  $ by eq. (\ref{muhanshu}) reads%
\begin{equation}
\sum_{E}P_{2}\left(  E;4\right)  z^{E}=\frac{\left(  1+z^{2}-z^{4}\right)
z^{6}}{\left(  z-1\right)  ^{4}\left(  z+1\right)  ^{2}\left(  1+z^{2}\right)
\left(  1+z+z^{2}\right)  }. \label{422}%
\end{equation}
Substituting eq. (\ref{422})\ into eq. (\ref{inverseZ}) gives%
\begin{align}
P_{2}\left(  E;4\right)   &  =\frac{1}{288}\left[  2E^{3}+6E^{2}%
-105E+9E\left(  -1\right)  ^{E}+9\left(  -1\right)  ^{E}\right. \nonumber\\
&  \left.  +18\left(  -i\right)  ^{E}+18i^{E}+64U_{E}\left(  -\frac{1}%
{2}\right)  +179\right]  .
\end{align}
\qquad The generating function of $P_{3}\left(  E;4\right)  $ by eq.
(\ref{muhanshu}) reads%
\begin{equation}
\sum_{E}P_{3}\left(  E;4\right)  z^{E}=\frac{\left(  1+z-z^{3}-z^{4}%
+z^{5}\right)  z^{5}}{\left(  z-1\right)  ^{4}\left(  z+1\right)  ^{2}\left(
1+z^{2}\right)  \left(  1+z+z^{2}\right)  }. \label{433}%
\end{equation}
Substituting eq. (\ref{433})\ into eq. (\ref{inverseZ}) gives%
\begin{align}
P_{3}\left(  E;4\right)   &  =\frac{1}{288}\left[  2E^{3}+6E^{2}-9E+9\left(
-1\right)  ^{E}E\right. \nonumber\\
&  \left.  -63\left(  -1\right)  ^{E}-54\left(  -i\right)  ^{E}-54i^{E}%
-32U_{E}\left(  -\frac{1}{2}\right)  -85\right]  .
\end{align}

$N=5$. The generating function of $Q\left(  E;5\right)  $ by eq.
(\ref{muhanshu2}) reads%
\begin{equation}
\sum_{E}Q\left(  E;5\right)  z^{E}=-\frac{z^{15}}{\left(  z-1\right)
^{5}\left(  1+z\right)  ^{2}\left(  1+z^{2}\right)  \left(  1+z+z^{2}\right)
\left(  1+z+z^{2}+z^{3}+z^{4}\right)  }. \label{51}%
\end{equation}
Substituting eq. (\ref{51})\ into eq. (\ref{inverseZ}) gives%
\begin{align}
P_{1}\left(  E;5\right)   &  =\frac{1}{86400}\left[  30E^{4}-900E^{3}%
+9300E^{2}-38250E+6912\cos\left(  \frac{2}{5}E\pi\right)  \right. \nonumber\\
&  +5400\cos\left(  \frac{E}{2}\pi\right)  +6400\cos\left(  \frac{2}{3}%
E\pi\right)  +6912\cos\left(  \frac{4}{5}E\pi\right)  +10125\cos\left(
E\pi\right) \nonumber\\
&  \left.  -1350\cos\left(  E\pi\right)  E-5400\sin\left(  \frac{E}{2}%
\pi\right)  +10125i\sin\left(  E\pi\right)  -1350i\sin\left(  E\pi\right)
E+50651\right]  . \label{333}%
\end{align}
The generating function of $P_{2}\left(  E;5\right)  $ by eq. (\ref{muhanshu})
reads%
\begin{equation}
\sum_{E}P_{2}\left(  E;5\right)  z^{E}=\frac{\left(  -1-z+z^{5}\right)  z^{9}%
}{\left(  z-1\right)  ^{5}\left(  z+1\right)  ^{2}\left(  z^{2}+1\right)
\left(  z^{2}+z+1\right)  \left(  1+z+z^{2}+z^{3}+z^{4}\right)  }. \label{52}%
\end{equation}
Substituting eq. (\ref{52})\ into eq. (\ref{inverseZ}) gives%
\begin{align}
P_{2}\left(  E;5\right)   &  =\frac{1}{86400}\left[  30E^{4}+300E^{3}%
-6900E^{2}+26550E-1350\left(  -1\right)  ^{E}E\right. \nonumber\\
&  +7425\left(  -1\right)  ^{E}+6912\cos\left(  \frac{2}{5}E\pi\right)
-5400\cos\left(  \frac{E}{2}\pi\right)  +6400\cos\left(  \frac{2}{3}%
E\pi\right) \nonumber\\
&  \left.  +6912\cos\left(  \frac{4}{5}E\pi\right)  +5400\sin\left(  \frac
{E}{2}\pi\right)  +6400\sqrt{3}\sin\left(  \frac{2}{3}E\pi\right)
-22249\right]  .
\end{align}
The generating function for $P_{3}\left(  E;5\right)  $, eq. (\ref{muhanshu}),
reads%
\begin{equation}
\sum_{E}P_{3}\left(  E;5\right)  z^{E}=\frac{\left(  -1-z+z^{5}+z^{6}%
-z^{8}\right)  z^{7}}{\left(  z-1\right)  ^{5}\left(  z+1\right)  ^{2}\left(
z^{2}+1\right)  \left(  z^{2}+z+1\right)  \left(  1+z+z^{2}+z^{3}%
+z^{4}\right)  }. \label{53}%
\end{equation}
Substituting eq. (\ref{53})\ into eq. (\ref{inverseZ}) gives%
\begin{align}
P_{3}\left(  E;5\right)   &  =\frac{1}{86400}\left[  30E^{4}+300E^{3}%
+300E^{2}-23850E+6912\cos\left(  \frac{2\pi}{5}E\right)  +16200\cos\left(
\frac{\pi}{2}E\right)  \right. \nonumber\\
&  -3200\cos\left(  \frac{2\pi}{3}E\right)  +6912\cos\left(  \frac{4}{5}%
E\pi\right)  +7425\cos(\pi E)-1350\cos\left(  \pi E\right)  E\nonumber\\
&  \left.  -16200\sin\left(  \frac{\pi}{2}E\right)  -3200\sqrt{3}\sin\left(
\frac{2\pi}{3}E\right)  +7425i\sin\left(  \pi E\right)  -1350iE\sin\left(  \pi
E\right)  +52151\right]  . \label{222}%
\end{align}
The generating function of $P_{4}\left(  E;5\right)  $ by eq. (\ref{muhanshu})
reads%
\begin{equation}
\sum_{E}P_{4}\left(  E;5\right)  z^{E}=\frac{\left(  -1-z+2z^{4}-z^{7}%
-z^{8}+z^{9}\right)  z^{6}}{\left(  z-1\right)  ^{5}\left(  z+1\right)
^{2}\left(  z^{2}+1\right)  \left(  z^{2}+z+1\right)  \left(  1+z+z^{2}%
+z^{3}+z^{4}\right)  }. \label{54}%
\end{equation}
Substituting eq. (\ref{54})\ into eq. (\ref{inverseZ}) gives%
\begin{align}
P_{4}\left(  E;5\right)   &  =\frac{1}{86400}\left[  30E^{4}+300E^{3}%
+300E^{2}-2250E+5400\sin\left(  \frac{\pi}{2}E\right)  -3200\sqrt{3}%
\sin\left(  \frac{2\pi}{3}E\right)  \right. \nonumber\\
&  -675i\left(  2E+5\right)  \sin\left(  \pi E\right)  -5400\cos\left(
\frac{\pi}{2}E\right)  -3200\cos\left(  \frac{2\pi}{3}E\right) \nonumber\\
&  -27648\cos\left(  \frac{2\pi}{5}E\right)  \left.  -27648\cos\left(
\frac{4\pi}{5}E\right)  -675\left(  2E+5\right)  \cos\left(  \pi E\right)
-19129\right]  .
\end{align}

Moreover, we also calculate the expression of $P\left(  E,N\right)  $ by the
generating function. In Ref.
\cite{andrews1998theory,andrews2004integer,gupta1970partitions}, the
expression of $P\left(  E,N\right)  $ is obtained by the recursive method.

The generating function of $P\left(  E;2\right)  $\ by eq. (\ref{muhanshu1})
reads
\begin{equation}
\sum_{E}P\left(  E;2\right)  z^{E}=\frac{z^{2}}{\left(  z-1\right)
^{2}\left(  1+z\right)  }. \label{22}%
\end{equation}
Substituting eq. (\ref{22})\ into eq. (\ref{inverseZ}) gives%
\begin{equation}
P\left(  E;2\right)  =\frac{E}{2}+\frac{1}{4}\left(  -1\right)  ^{E}-\frac
{1}{4}. \label{P(E;2)}%
\end{equation}
The generating function of $P\left(  E;3\right)  $ by eq. (\ref{muhanshu1})
reads%
\begin{equation}
\sum_{E}P\left(  E;3\right)  z^{E}=-\frac{z^{3}}{\left(  z-1\right)
^{3}\left(  1+z\right)  \left(  1+z+z^{2}\right)  }. \label{33}%
\end{equation}
Substituting eq. (\ref{33})\ into eq. (\ref{inverseZ}) gives%
\begin{equation}
P\left(  E;3\right)  =\frac{1}{72}\left[  6E^{2}+16\cos\left(  \frac{2E\pi}%
{3}\right)  -9\left(  -1\right)  ^{E}-7\right]  . \label{P(E;3)}%
\end{equation}
The generating function of $P\left(  E;4\right)  $ by eq. (\ref{muhanshu1})
reads%
\begin{equation}
\sum_{E}P\left(  E;4\right)  z^{E}=\frac{z^{4}}{\left(  z-1\right)
^{4}\left(  1+z\right)  ^{2}\left(  1+z^{2}\right)  \left(  1+z+z^{2}\right)
}. \label{44}%
\end{equation}
Substituting eq. (\ref{44})\ into eq. (\ref{inverseZ}) gives%
\begin{align}
P\left(  E;4\right)   &  =\frac{1}{288}\left[  2E^{3}+6E^{2}-9E+9\left(
-1\right)  ^{E}E\right. \nonumber\\
&  \left.  +9\left(  -1\right)  ^{E}+18\left(  -i\right)  ^{E}+18i^{E}%
-32U_{E}\left(  -\frac{1}{2}\right)  -13\right]  . \label{P(E;4)}%
\end{align}
The generating function of $P\left(  E;5\right)  $ by eq. (\ref{muhanshu})
reads%
\begin{equation}
\sum_{E}P\left(  E;5\right)  z^{E}=-\frac{z^{5}}{\left(  z-1\right)
^{5}\left(  1+z\right)  ^{2}\left(  1+z^{2}\right)  \left(  1+z+z^{2}\right)
\left(  1+z+z^{2}+z^{3}+z^{4}\right)  }. \label{55}%
\end{equation}
Substituting eq. (\ref{55})\ into eq. (\ref{inverseZ}) gives%
\begin{align}
P\left(  E;5\right)   &  =\frac{1}{86400}\left[  30E^{4}+300E^{3}%
+300E^{2}-2250E-1350E\left(  -1\right)  ^{E}+3375\left(  -1\right)
^{E}\right. \nonumber\\
&  +5400\sin\left(  \frac{E}{2}\pi\right)  -3200\sqrt{3}\sin\left(  \frac
{2}{3}E\pi\right)  -5400\cos\left(  \frac{E}{2}\pi\right) \nonumber\\
&  \left.  -3200\cos\left(  \frac{2}{3}E\pi\right)  +6912\cos\left(  \frac
{2}{5}E\pi\right)  +6912\cos\left(  \frac{4}{5}E\pi\right)  -1849\right]  .
\label{P(E;5)}%
\end{align}

It is worthy to note that here, though there are imaginary units in the
expressions, the result are real. To illustrate this, we list some explicit
results of $P\left(  E;N\right)  $, $Q\left(  E;N\right)  $, and $P_{q}\left(
E;N\right)  $ for different $E$ in table \ref{table1}.

\begin{table}[ptb]
\caption{Some explicit results of $P\left(  E;N\right)  $, $Q\left(
E;N\right)  $, and $P_{q}\left(  E;N\right)  $ for different $E$.}%
\label{table1}%
\centering
\begin{tabular}
[c]{ccccccccccc}\hline
$E$ & $10$ & $20$ & $30$ & $40$ & $50$ & $60$ & $70$ & $80$ & $90$ &
$100$\\\hline
$Q\left(  E;2\right)  $ & $4$ & $9$ & $14$ & $19$ & $24$ & $29$ & $34$ & $39$
& $44$ & $49$\\
$P\left(  E;2\right)  $ & $5$ & $10$ & $15$ & $20$ & $25$ & $30$ & $35$ & $40$
& $45$ & $50$\\
$Q\left(  E;3\right)  $ & $4$ & $24$ & $61$ & $114$ & $184$ & $271$ & $374$ &
$494$ & $631$ & $784$\\
$P_{2}\left(  E;3\right)  $ & $8$ & $33$ & $74$ & $133$ & $208$ & $299$ &
$408$ & $533$ & $674$ & $833$\\
$P\left(  E;3\right)  $ & $8$ & $33$ & $75$ & $133$ & $208$ & $300$ & $408$ &
$533$ & $675$ & $833$\\
$Q\left(  E;4\right)  $ & $1$ & $23$ & $108$ & $297$ & $632$ & $1154$ & $1906$
& $2928$ & $4263$ & $5952$\\
$P_{2}\left(  E;4\right)  $ & $6$ & $58$ & $197$ & $465$ & $904$ & $1556$ &
$2461$ & $3663$ & $5202$ & $7120$\\
$P_{3}\left(  E;4\right)  $ & $9$ & $63$ & $206$ & $477$ & $920$ & $1574$ &
$2484$ & $3688$ & $5231$ & $7152$\\
$P\left(  E;4\right)  $ & $9$ & $64$ & $206$ & $487$ & $920$ & $1575$ & $2484$
& $3689$ & $5231$ & $7153$\\
$Q\left(  E;5\right)  $ & $0$ & $7$ & $84$ & $377$ & $1115$ & $2611$ & $5260$
& $9542$ & $16019$ & $25337$\\
$P_{2}\left(  E;5\right)  $ & $2$ & $57$ & $312$ & $995$ & $2419$ & $4980$ &
$9157$ & $15512$ & $24692$ & $37425$\\
$P_{3}\left(  E;5\right)  $ & $5$ & $80$ & $370$ & $1106$ & $2599$ & $5246$ &
$9525$ & $16000$ & $25315$ & $38201$\\
$P_{4}\left(  E;5\right)  $ & $6$ & $83$ & $376$ & $1114$ & $2610$ & $5259$ &
$9541$ & $16018$ & $25336$ & $38224$\\
$P\left(  E;5\right)  $ & $7 $ & $84$ & $377$ & $1115$ & $2611$ & $5260$ &
$9542$ & $16019$ & $25337$ & $38225$\\\hline
\end{tabular}
\end{table}

The results of $P\left(  E,N\right)  $ given by eqs. (\ref{P(E;2)}),
(\ref{P(E;3)}), (\ref{P(E;4)}), and (\ref{P(E;5)}) coincide\textbf{ }with the
results obtained by the recursive method
\cite{andrews1998theory,andrews2004integer,gupta1970partitions}.

\section{Conclusions \label{con}}

In this paper, starting from the canonical partition function of various kinds
of quantum ideal gases, we obtain the generating function of some restricted
integer partition functions that count the number of integer partitions with
length $N$. We calculate the exact expression of restricted integer partition
functions from the corresponding generating functions which are constructed by
resorting to statistical mechanics.

We introduce a new type of restricted integer partition functions. The
restricted integer partition functions introduced in the present paper
corresponds to general statistics which is a generalization of Gentile
statistics proposed in Ref. \cite{dai2009exactly}. Many kinds of integer
partition functions are special cases of this restricted integer partition
function. This allows us to consider a number of restricted partition
functions in a unified framework.

We also obtain a relation between the integer partition function and the
symmetric function. Concretely, we show that the generating function of the
restricted integer partition functions corresponding to ideal Bose, Fermi, and
Gentile gases are symmetric functions and can be expressed as linear
combinations of the $S$-function which is an important class of the symmetric function.

We also provides some expressions of restricted integer partition functions,
by use of the approach suggested in the paper, as examples.

The generating function of the restricted integer partition functions obtained
in the present paper is obtained from canonical partition functions which is
from the canonical ensemble. It is in principle possible to calculate various
restricted integer partition functions from other statistical ensembles. For
example, the thermodynamic quantities calculated in\ the canonical ensemble
can also be obtained in the grand canonical ensemble to some extent
\cite{dai2017explicit,leonard1968exact}. Furthermore, the canonical partition
function is nothing but a kind of spectral functions which are functions of
the eigenvalue of a system and various spectral functions can be achieved from
each other \cite{dai2009number,vassilevich2003heat,dai2010approach}. This in
principle allows us to start with other statistical ensembles and
thermodynamic quantities. All the information of a thermodynamic system is
embodied in the corresponding mechanical system
\cite{hoyuelos2018creation,dai2012calculating}, which inspires us to find
relations between partitions and mechanics. Moreover, the result provided in
the present paper may be useful in the additive and diophantine problems
\cite{hardy1916some,schmidt1991diophantine,duverney2010number,mitrinovic1996handbook}%

\appendix
\titleformat{\section}{\large\bfseries}{\textbf{Appendix \thesection}}{1em}{}
\titlecontents{section}[1.08em]{\vspace{.5\baselineskip}}%
             { \textbf{Appendix} \textbf{\thecontentslabel} \qquad}{} %
             {\hspace{.5em}\titlerule*[10pt]{$\cdot$}\contentspage}

\section{The expression of the generating function of $P^{\left\{
\varepsilon\right\}  }\left(  E;N\right)  $, $Q^{\left\{  \varepsilon\right\}
}\left(  E;N\right)  $, and $P_{q}^{\left\{  \varepsilon\right\}  }\left(
E;N\right)  $ \label{QIq}}

In this appendix, we express the generating function of $P^{\left\{
\varepsilon\right\}  }\left(  E;N\right)  $ and $Q^{\left\{  \varepsilon
\right\}  }\left(  E;N\right)  $ in terms of the determinant of certain
matrices and list the explicit expressions of the generating function of
$P_{q}^{\left\{  \varepsilon\right\}  }\left(  E;N\right)  $ for $N=3$, $4$,
$5$, and $6$. The details of the calculation can be found in Ref.
\cite{zhou2017canonical}.

\subsection{The matrix form\ of expressions of $P^{\left\{  \varepsilon
\right\}  }\left(  E;N\right)  $ and $Q^{\left\{  \varepsilon\right\}
}\left(  E;N\right)  $}

The generating functions of $P^{\left\{  \varepsilon\right\}  }\left(
E;N\right)  $ and $Q^{\left\{  \varepsilon\right\}  }\left(  E;N\right)  $ are
given in eqs. (\ref{muhanshu1}) and (\ref{muhanshu}) in terms of
$S$-functions. However, the $S$-functions $\left(  N\right)  \left(
x_{1},x_{2},\ldots\right)  $ and $\left(  1^{N}\right)  \left(  x_{1}%
,x_{2},\ldots\right)  $ can be represented as the determinant of a matrix
\cite{littlewood1977theory,macdonald1998symmetric}. In this section, we
express the generating function of $P^{\left\{  \varepsilon\right\}  }\left(
E;N\right)  $ and $Q^{\left\{  \varepsilon\right\}  }\left(  E;N\right)  $ in
matrices forms.
\begin{equation}
\sum_{E}Q^{\left\{  \varepsilon\right\}  }\left(  E,N\right)  z^{E}=\frac
{1}{N!}\det\left(
\begin{array}
[c]{ccccc}%
p\left(  z\right)  & 1 & 0 & \ldots & 0\\
p\left(  z^{2}\right)  & p\left(  z\right)  & 2 & \ldots & 0\\
p\left(  z^{3}\right)  & p\left(  z^{2}\right)  & p\left(  z\right)  & \ldots
& 0\\
\ldots & \ldots & \ldots & \ldots & N-1\\
p\left(  z^{N}\right)  & p\left(  z^{N-1}\right)  & \ldots & \ldots & p\left(
z\right)
\end{array}
\right)  ,
\end{equation}
where $p\left(  z\right)  =\sum_{\varepsilon\in\left\{  \varepsilon\right\}
}z^{\varepsilon}$.%
\begin{equation}
\sum_{E}P^{\left\{  \varepsilon\right\}  }\left(  E,N\right)  z^{E}=\frac
{1}{N!}\det\left(
\begin{array}
[c]{ccccc}%
p\left(  z\right)  & -1 & 0 & \ldots & 0\\
p\left(  z^{2}\right)  & p\left(  z\right)  & -2 & \ldots & 0\\
p\left(  z^{3}\right)  & p\left(  z^{2}\right)  & p\left(  z\right)  & \ldots
& 0\\
\ldots & \ldots & \ldots & \ldots & -\left(  N-1\right) \\
p\left(  z^{N}\right)  & p\left(  z^{N-1}\right)  & \ldots & \ldots & p\left(
z\right)
\end{array}
\right)  .
\end{equation}

\subsection{The explicit expression of $P_{q}^{\left\{  \varepsilon\right\}
}\left(  E;N\right)  $ for $N=3$, $4$, $5$, and $6$}

In this section, we list the explicit expressions of $P_{q}^{\left\{
\varepsilon\right\}  }\left(  E;N\right)  $ for $N=3$, $4$, $5$, and $6$. The
details of the calculation can be found in Ref. \cite{zhou2017canonical}.%

\begin{equation}
\sum_{E}P_{2}^{\left\{  \varepsilon\right\}  }\left(  E,3\right)  z^{E}%
=\frac{1}{6}p\left(  z\right)  ^{3}+\frac{1}{2}p\left(  z\right)  p\left(
z^{2}\right)  -\frac{2}{3}p\left(  z^{3}\right)  .
\end{equation}

\begin{equation}
\sum_{E}P_{2}^{\left\{  \varepsilon\right\}  }\left(  E,4\right)  z^{E}%
=\frac{1}{24}p\left(  z\right)  ^{4}+\frac{1}{4}p\left(  z\right)
^{2}p\left(  z^{2}\right)  +\frac{1}{8}p\left(  z^{2}\right)  ^{2}-\frac{2}%
{3}p\left(  z^{3}\right)  p\left(  z\right)  +\frac{1}{4}p\left(
z^{4}\right)  .
\end{equation}
\qquad%
\begin{equation}
\sum_{E}P_{3}^{\left\{  \varepsilon\right\}  }\left(  E,4\right)  z^{E}%
=\frac{1}{24}p\left(  z\right)  ^{4}+\frac{1}{4}p\left(  z\right)
^{2}p\left(  z^{2}\right)  +\frac{1}{8}p\left(  z^{2}\right)  ^{2}+\frac{1}%
{3}p\left(  z^{3}\right)  p\left(  z\right)  -\frac{3}{4}p\left(
z^{4}\right)  .
\end{equation}%
\begin{align}
\sum_{E}P_{2}^{\left\{  \varepsilon\right\}  }\left(  E,5\right)  z^{E}  &
=\frac{1}{120}p\left(  z\right)  ^{5}+\frac{1}{12}p\left(  z\right)
^{3}p\left(  z^{2}\right)  +\frac{1}{8}p\left(  z\right)  p\left(
z^{2}\right)  ^{2}\nonumber\\
&  -\frac{1}{3}p\left(  z\right)  ^{2}p\left(  z^{3}\right)  -\frac{1}%
{3}p\left(  z^{2}\right)  p\left(  z^{3}\right)  +\frac{1}{4}p\left(
z\right)  p\left(  z^{4}\right)  +\frac{1}{5}p\left(  z^{5}\right)  .
\end{align}%
\begin{align}
\sum_{E}P_{3}^{\left\{  \varepsilon\right\}  }\left(  E,5\right)  z^{E}  &
=\frac{1}{120}p\left(  z\right)  ^{5}+\frac{1}{12}p\left(  z\right)
^{3}p\left(  z^{2}\right)  +\frac{1}{8}p\left(  z\right)  p\left(
z^{2}\right)  ^{2}\nonumber\\
&  +\frac{1}{6}p\left(  z\right)  ^{2}p\left(  z^{3}\right)  +\frac{1}%
{6}p\left(  z^{2}\right)  p\left(  z^{3}\right)  -\frac{3}{4}p\left(
z\right)  p\left(  z^{4}\right)  +\frac{1}{5}p\left(  z^{5}\right)  .
\end{align}%
\begin{align}
\sum_{E}P_{4}^{\left\{  \varepsilon\right\}  }\left(  E,5\right)  z^{E}  &
=\frac{1}{120}p\left(  z\right)  ^{5}+\frac{1}{12}p\left(  z\right)
^{3}p\left(  z^{2}\right)  +\frac{1}{8}p\left(  z\right)  p\left(
z^{2}\right)  ^{2}\nonumber\\
&  +\frac{1}{6}p\left(  z\right)  ^{2}p\left(  z^{3}\right)  +\frac{1}%
{6}p\left(  z^{2}\right)  p\left(  z^{3}\right)  +\frac{1}{4}p\left(
z\right)  p\left(  z^{4}\right)  -\frac{4}{5}p\left(  z^{5}\right)  .
\end{align}

\begin{align}
\sum_{E}P_{2}^{\left\{  \varepsilon\right\}  }\left(  E,6\right)  z^{E}  &
=\frac{1}{6!}p\left(  z\right)  ^{6}+\frac{1}{48}p\left(  z\right)
^{4}p\left(  z^{2}\right)  +\frac{1}{16}p\left(  z\right)  ^{2}p\left(
z^{2}\right)  ^{2}+\frac{1}{48}p\left(  z^{2}\right)  ^{3}\nonumber\\
&  -\frac{1}{9}p\left(  z\right)  ^{3}p\left(  z^{3}\right)  -\frac{1}%
{3}p\left(  z\right)  p\left(  z^{2}\right)  p\left(  z^{3}\right)  +\frac
{2}{9}p\left(  z^{3}\right)  ^{2}\nonumber\\
&  +\frac{1}{8}p\left(  z\right)  ^{2}p\left(  z^{4}\right)  +\frac{1}%
{8}p\left(  z^{2}\right)  p\left(  z^{4}\right)  +\frac{1}{5}p\left(
z^{5}\right)  p\left(  z\right)  -\frac{1}{3}p\left(  z^{6}\right)  ,
\label{i}%
\end{align}%
\begin{align}
\sum_{E}P_{3}^{\left\{  \varepsilon\right\}  }\left(  E,6\right)  z^{E}  &
=\frac{1}{6!}p\left(  z\right)  ^{6}+\frac{1}{48}p\left(  z\right)
^{4}p\left(  z^{2}\right)  +\frac{1}{16}p\left(  z\right)  ^{2}p\left(
z^{2}\right)  ^{2}+\frac{1}{48}p\left(  z^{2}\right)  ^{3}\nonumber\\
&  +\frac{1}{18}p\left(  z\right)  ^{3}p\left(  z^{3}\right)  +\frac{1}%
{6}p\left(  z\right)  p\left(  z^{2}\right)  p\left(  z^{3}\right)  +\frac
{1}{18}p\left(  z^{3}\right)  ^{2}\nonumber\\
&  -\frac{3}{8}p\left(  z\right)  ^{2}p\left(  z^{4}\right)  -\frac{3}%
{8}p\left(  z^{2}\right)  p\left(  z^{4}\right)  +\frac{1}{5}p\left(
z^{5}\right)  p\left(  z\right)  +\frac{1}{6}p\left(  z^{6}\right)  ,
\label{j}%
\end{align}%
\begin{align}
\sum_{E}P_{4}^{\left\{  \varepsilon\right\}  }\left(  E,6\right)  z^{E}  &
=\frac{1}{6!}p\left(  z\right)  ^{6}+\frac{1}{48}p\left(  z\right)
^{4}p\left(  z^{2}\right)  +\frac{1}{16}p\left(  z\right)  ^{2}p\left(
z^{2}\right)  ^{2}+\frac{1}{48}p\left(  z^{2}\right)  ^{3}\nonumber\\
&  +\frac{1}{18}p\left(  z\right)  ^{3}p\left(  z^{3}\right)  +\frac{1}%
{6}p\left(  z\right)  p\left(  z^{2}\right)  p\left(  z^{3}\right)  +\frac
{1}{18}p\left(  z^{3}\right)  ^{2}\nonumber\\
&  +\frac{1}{8}p\left(  z\right)  ^{2}p\left(  z^{4}\right)  +\frac{1}%
{8}p\left(  z^{2}\right)  p\left(  z^{4}\right)  -\frac{4}{5}p\left(
z^{5}\right)  p\left(  z\right)  +\frac{1}{6}p\left(  z^{6}\right)  ,
\label{k}%
\end{align}%
\begin{align}
\sum_{E}P_{5}^{\left\{  \varepsilon\right\}  }\left(  E,6\right)  z^{E}  &
=\frac{1}{6!}p\left(  z\right)  ^{6}+\frac{1}{48}p\left(  z\right)
^{4}p\left(  z^{2}\right)  +\frac{1}{16}p\left(  z\right)  ^{2}p\left(
z^{2}\right)  ^{2}+\frac{1}{48}p\left(  z^{2}\right)  ^{3}\nonumber\\
&  +\frac{1}{18}p\left(  z\right)  ^{3}p\left(  z^{3}\right)  +\frac{1}%
{6}p\left(  z\right)  p\left(  z^{2}\right)  p\left(  z^{3}\right)  +\frac
{1}{18}p\left(  z^{3}\right)  ^{2}\nonumber\\
&  +\frac{1}{8}p\left(  z\right)  ^{2}p\left(  z^{4}\right)  +\frac{1}%
{8}p\left(  z^{2}\right)  p\left(  z^{4}\right)  +\frac{1}{5}p\left(
z^{5}\right)  p\left(  z\right)  -\frac{5}{6}p\left(  z^{6}\right)  .
\label{l}%
\end{align}


\acknowledgments

We are very indebted to Dr G. Zeitrauman for his encouragement. This work is supported in part by NSF of China under Grant
No. 11575125 and No. 11675119.










\providecommand{\href}[2]{#2}\begingroup\raggedright\endgroup


\begin{thebibliography}{10}

\bibitem{andrews1998theory}
G.~Andrews, {\em The Theory of Partitions}.
\newblock Cambridge mathematical library. Cambridge University Press, 1998.

\bibitem{andrews2004integer}
G.~E. Andrews and K.~Eriksson, {\em Integer partitions}.
\newblock Cambridge University Press, 2004.

\bibitem{hardy1999ramanujan}
G.~Hardy, {\em Ramanujan: Twelve Lectures on Subjects Suggested by His Life and
  Work}.
\newblock AMS Chelsea Publishing Series. AMS Chelr={IOP Publishing}, 1999.

\bibitem{pathria2011statistical}
R.~Pathria, {\em Statistical Mechanics}.
\newblock Elsevier Science, 2011.

\bibitem{reichl2009modern}
L.~Reichl, {\em A Modern Course in Statistical Physics}.
\newblock Physics textbook. Wiley, 2009.

\bibitem{gentile1940itosservazioni}
G.~Gentile~jr, {\it Osservazioni sopra le statistiche intermedie},  {\em Il
  Nuovo Cimento (1924-1942)} {\bf 17} (1940), no.~10 493--497.

\bibitem{dai2004gentile}
W.-S. Dai and M.~Xie, {\it Gentile statistics with a large maximum occupation
  number},  {\em Annals of Physics} {\bf 309} (2004), no.~2 295--305.

\bibitem{dai2004representation}
W.-S. Dai and M.~Xie, {\it A representation of angular momentum (SU (2))
  algebra},  {\em Physica A: Statistical Mechanics and its Applications} {\bf
  331} (2004), no.~3-4 497--504.

\bibitem{maslov2017relationship}
V.~P. Maslov, {\it The relationship between the Fermi--Dirac distribution and
  statistical distributions in languages},  {\em Mathematical Notes} {\bf 101}
  (2017), no.~3-4 645--659.

\bibitem{shen2007intermediate}
Y.~Shen, W.-S. Dai, and M.~Xie, {\it Intermediate-statistics quantum bracket,
  coherent state, oscillator, and representation of angular momentum [su (2)]
  algebra},  {\em Physical Review A} {\bf 75} (2007), no.~4 042111.

\bibitem{maslov2017model}
V.~Maslov, {\it A model of classical thermodynamics based on the partition
  theory of integers, Earth gravitation, and semiclassical asymptotics. I},
  {\em Russian Journal of Mathematical Physics} {\bf 24} (2017), no.~3
  354--372.

\bibitem{dai2009intermediate}
W.-S. Dai and M.~Xie, {\it Intermediate-statistics spin waves},  {\em Journal
  of Statistical Mechanics: Theory and Experiment} {\bf 2009} (2009), no.~04
  P04021.

\bibitem{dai2009exactly}
W.-S. Dai and M.~Xie, {\it An exactly solvable phase transition model:
  generalized statistics and generalized Bose--Einstein condensation},  {\em
  Journal of Statistical Mechanics: Theory and Experiment} {\bf 2009} (2009),
  no.~07 P07034.

\bibitem{zhou2017canonical}
C.-C. Zhou and W.-S. Dai, {\it Canonical partition functions: ideal quantum
  gases, interacting classical gases, and interacting quantum gases},  {\em
  Journal of Statistical Mechanics: Theory and Experiment} {\bf 2018} (2018),
  no.~2 023105.

\bibitem{littlewood1977theory}
D.~E. Littlewood, {\em The theory of group characters and matrix
  representations of groups}, vol.~357.
\newblock American Mathematical Soc., 1977.

\bibitem{macdonald1998symmetric}
I.~G. Macdonald, {\em Symmetric functions and Hall polynomials}.
\newblock Oxford university press, 1998.

\bibitem{van1937statistical}
C.~Van~Lier and G.~Uhlenbeck, {\it On the statistical calculation of the
  density of the energy levels of the nuclei},  {\em Physica} {\bf 4} (1937),
  no.~7 531--542.

\bibitem{auluck1946statistical}
F.~Auluck and D.~Kothari, {\it Statistical mechanics and the partitions of
  numbers},  in {\em Mathematical Proceedings of the Cambridge Philosophical
  Society}, vol.~42, pp.~272--277, Cambridge Univ Press, 1946.

\bibitem{tran2004quantum}
M.~N. Tran, M.~Murthy, and R.~K. Bhaduri, {\it On the quantum density of states
  and partitioning an integer},  {\em Annals of Physics} {\bf 311} (2004),
  no.~1 204--219.

\bibitem{kubasiak2005fermi}
A.~Kubasiak, J.~K. Korbicz, J.~Zakrzewski, and M.~Lewenstein, {\it Fermi-Dirac
  statistics and the number theory},  {\em EPL (Europhysics Letters)} {\bf 72}
  (2005), no.~4 506.

\bibitem{srivatsan2006gentile}
C.~Srivatsan, M.~V. Murthy, and R.~Bhaduri, {\it Gentile statistics and
  restricted partitions},  {\em Pramana} {\bf 66} (2006), no.~3 485--494.

\bibitem{bogoliubov2007enumeration}
N.~Bogoliubov, {\it Enumeration of plane partitions and the algebraic Bethe
  anzatz},  {\em Theoretical and Mathematical Physics} {\bf 150} (2007), no.~2
  165--174.

\bibitem{prokhorov2012asymptotic}
D.~Prokhorov and A.~Rovenchak, {\it Asymptotic formulas for integer partitions
  within the approach of microcanonical ensemble.},  {\em Condensed Matter
  Physics} {\bf 15} (2012), no.~3.

\bibitem{rovenchak2014enumeration}
A.~Rovenchak, {\it Enumeration of plane partitions with a restricted number of
  parts},  {\em Theoretical and Mathematical Physics} {\bf 181} (2014), no.~2
  1428--1434.

\bibitem{rovenchak2016statistical}
A.~Rovenchak, {\it Statistical mechanics approach in the counting of integer
  partitions},  {\em Banach Center Publications} {\bf 109} (2016) 155--166.

\bibitem{maslov2017new}
V.~P. Maslov, {\it New insight into the partition theory of integers related to
  problems of thermodynamics and mesoscopic physics},  {\em Mathematical Notes}
  {\bf 102} (2017), no.~1-2 232--249.

\bibitem{maslov2017topological}
V.~Maslov, {\it Topological phase transitions in the theory of partitions of
  integers},  {\em Russian Journal of Mathematical Physics} {\bf 24} (2017),
  no.~2 249--260.

\bibitem{grossman1999number}
S.~Grossman and M.~Holthaus, {\it From number theory to statistical mechanics:
  Bose-Einstein condensation in isolated traps},  {\em Chaos Solitons and
  Fractals} {\bf 10} (1999), no.~4 795--804.

\bibitem{chaturvedi1996canonical}
S.~Chaturvedi, {\it Canonical partition functions for parastatistical systems
  of any order},  {\em Physical Review E} {\bf 54} (1996), no.~2 1378.

\bibitem{balantekin2001partition}
A.~Balantekin, {\it Partition functions in statistical mechanics, symmetric
  functions, and group representations},  {\em Physical Review E} {\bf 64}
  no.~6 066105.

\bibitem{mcnamara2009factorial}
P.~J. McNamara, {\it Factorial Schur functions via the six vertex model},  {\em
  Citeseer} (2009).

\bibitem{schmidt2002partition}
H.-J. Schmidt and J.~Schnack, {\it Partition functions and symmetric
  polynomials},  {\em American Journal of Physics} {\bf 70} (2002), no.~1
  53--57.

\bibitem{gorin2015asymptotics}
V.~Gorin, G.~Panova, et~al., {\it Asymptotics of symmetric polynomials with
  applications to statistical mechanics and representation theory},  {\em The
  Annals of Probability} {\bf 43} (2015), no.~6 3052--3132.

\bibitem{stanley1971theory}
R.~P. Stanley, {\it Theory and application of plane partitions: Part 1},  {\em
  Studies in Applied Mathematics} {\bf 50} (1971), no.~2 167--188.

\bibitem{stanley1971theory1}
R.~P. Stanley, {\it Theory and application of plane partitions. Part 2},  {\em
  Studies in Applied Mathematics} {\bf 50} (1971), no.~3 259--279.

\bibitem{fel2016gaussian}
L.~G. Fel, {\it Gaussian Polynomials and Restricted Partition Functions with
  Constraints},  {\em arXiv preprint arXiv:1611.09931} (2016).

\bibitem{mitrinovic1996handbook}
D.~S. Mitrinovic, J.~S{\'a}ndor, and B.~Crstici, {\em Handbook of number
  theory}.
\newblock Kluwer, 1996.

\bibitem{richard1999enumerative}
P.~S. Richard, {\it Enumerative Combinatorics},  {\em Volume I} (1999).

\bibitem{nathanson2000elementary}
M.~Nathanson, {\it Elementary Methods in Number Theory (Graduate Texts in
  Mathematics 195)},  2000.

\bibitem{euler1748introductio}
L.~Eulero, {\em Introductio in analysin infinitorum. Tomus primus.}
\newblock Lausannae: apud Marcum-Michaelem Bousquet \& Socios, 1748.

\bibitem{gupta1970partitions}
H.~Gupta, {\it Partitions--a survey},  {\em Journal of Res. of Nat. Bur.
  Standards-B Math. Sciences B} {\bf 74} (1970) 1--29.

\bibitem{erdos1941distribution}
P.~Erd{\"o}s, J.~Lehner, et~al., {\it The distribution of the number of
  summands in the partitions of a positive integer},  {\em Duke Mathematical
  Journal} {\bf 8} (1941), no.~2 335--345.

\bibitem{hamermesh1962group}
M.~Hamermesh, {\em Group theory and its application to physical problems}.
\newblock Courier Corporation, 1962.

\bibitem{akhiezerelements}
N.~Akhiezer and H.~McFaden, {\em Elements of the Theory of Elliptic Functions}.
\newblock Translations of Mathematical Monographs. American Mathematical Soc.

\bibitem{steuding2007value}
J.~Steuding, {\em Value-Distribution of L-Functions}.
\newblock Lecture Notes in Mathematics. Springer, 2007.

\bibitem{dette1995note}
H.~Dette, {\it A note on some peculiar nonlinear extremal phenomena of the
  Chebyshev polynomials},  {\em Proceedings of the Edinburgh Mathematical
  Society} {\bf 38} (1995), no.~2 343--355.

\bibitem{dai2017explicit}
W.-S. Dai and M.~Xie, {\it The explicit expression of the fugacity for weakly
  interacting Bose and Fermi gases},  {\em Journal of Mathematical Physics}
  {\bf 58} (2017), no.~11 113502.

\bibitem{leonard1968exact}
A.~Leonard, {\it Exact inversion of the fugacity-density relation for ideal
  quantum gases},  {\em Physical Review} {\bf 175} (1968), no.~1 221.

\bibitem{dai2009number}
W.-S. Dai and M.~Xie, {\it The number of eigenstates: counting function and
  heat kernel},  {\em Journal of High Energy Physics} {\bf 2009} (2009), no.~02
  033.

\bibitem{vassilevich2003heat}
D.~V. Vassilevich, {\it Heat kernel expansion: user's manual},  {\em Physics
  reports} {\bf 388} (2003), no.~5 279--360.

\bibitem{dai2010approach}
W.-S. Dai and M.~Xie, {\it An approach for the calculation of one-loop
  effective actions, vacuum energies, and spectral counting functions},  {\em
  Journal of High Energy Physics} {\bf 2010} (2010), no.~6 70.

\bibitem{hoyuelos2018creation}
M.~Hoyuelos, {\it From creation and annihilation operators to statistics},
  {\em Physica A: Statistical Mechanics and its Applications} {\bf 490} (2018)
  944--952.

\bibitem{dai2012calculating}
W.-S. Dai and M.~Xie, {\it Calculating statistical distributions from operator
  relations: The statistical distributions of various intermediate statistics},
   {\em Annals of Physics} {\bf 332} (2012) 166--179.

\bibitem{hardy1916some}
G.~Hardy and J.~Littlewood, {\it Some Problems of Diophantine Approximation A
  Remarkable Trigonometrical Series},  {\em Proceedings of the National Academy
  of Sciences} {\bf 2} (1916), no.~10 583--586.

\bibitem{schmidt1991diophantine}
W.~Schmidt, {\em Diophantine approximations and Diophantine equations}.
\newblock Lecture notes in mathematics. Springer-Verlag, 1991.

\bibitem{duverney2010number}
D.~Duverney, {\em Number Theory: An Elementary Introduction Through Diophantine
  Problems}.
\newblock Monographs in number theory. World Scientific, 2010.

\end{thebibliography}

\end{document}